%% file: main.tex
\documentclass[11pt]{article}
\input{settings/pkgs}

\input{settings/defs}
\input{settings/environments}

\input{settings/colors}
\usepackage[style=alphabetic,natbib=true,maxcitenames=2, maxbibnames=10,backend=bibtex]{biblatex}
\allowdisplaybreaks

\title{Single-dimensional Contract Design:\\ Efficient Algorithms and Learning}

\author{
    Martino Bernasconi\thanks{Bocconi University, \texttt{\textcolor{blueGrotto}{martino.bernasconi@unibocconi.it}}}\and
    Matteo Castiglioni\thanks{Politecnico di Milano, \texttt{\textcolor{blueGrotto}{matteo.castiglioni@polimi.it}}} \and
    Andrea Celli\thanks{Bocconi University, \texttt{\textcolor{blueGrotto}{andrea.celli2@unibocconi.it}}}
}

\date{}

\begin{document}

\maketitle
\input{src/abstract}

\pagenumbering{gobble} 

\clearpage
\tableofcontents
\renewcommand{\thmtformatoptarg}[1]{~#1}
\newpage
\pagenumbering{arabic}
\input{src/introduction}
\input{src/preliminaries}

\input{src/additivePTAS2}

\input{src/hardness}
\input{src/learning}

\input{src/conclusions}

\newpage
\appendix

\section*{Appendix}
\input{src/appendix}

\section*{Acknowledgments}

Funded by the European Union. Views and opinions expressed are however those of the author(s) only and do not necessarily reflect those of the European Union or the European Research Council Executive Agency. Neither the European Union nor the granting authority can be held responsible for them. 

\vspace{.5cm}
\noindent AC is supported by ERC grant (Project 101165466 — PLA-STEER). MB, MC are partially supported by the FAIR (Future Artificial Intelligence Research) project PE0000013, funded by the NextGenerationEU program within the PNRR-PE-AI scheme (M4C2, investment 1.3, line on Artificial Intelligence). MC is also partially supported by the EU Horizon project ELIAS (European Lighthouse of AI for Sustainability, No. 101120237).

\newpage
\printbibliography

\end{document}

%% file: settings/pkgs.tex
\usepackage{fullpage}

\usepackage{hyperref}
\usepackage[utf8]{inputenc} 
\usepackage[T1]{fontenc}    
\usepackage{url}            
\usepackage{mathtools}
\usepackage{booktabs}       
\usepackage{amsfonts}       
\usepackage{nicefrac}       
\usepackage{microtype}      
\usepackage{xcolor}         

\usepackage{nicematrix}

\usepackage{tikz}
\usepackage{float}
\usetikzlibrary{shapes,backgrounds,patterns,calc,arrows,arrows.meta}
\usepackage{chemfig}

\catcode`\_=11
\definearrow5{-x>}{%
    \CF_arrowshiftnodes{#3}%
    \CF_expafter{\draw[}\CF_arrowcurrentstyle](\CF_arrowstartnode)--(\CF_arrowendnode)%
        coordinate[midway,shift=(\CF_ifempty{#4}{225}{#4}+\CF_arrowcurrentangle:\CF_ifempty{#5}{5pt}{#5})](line@start)%
        coordinate[midway,shift=(\CF_ifempty{#4}{45}{#4}+\CF_arrowcurrentangle:\CF_ifempty{#5}{5pt}{#5})](line@end)%
        coordinate[midway,shift=(\CF_ifempty{#4}{135}{#4}+\CF_arrowcurrentangle:\CF_ifempty{#5}{5pt}{#5})](line@start@i)%
        coordinate[midway,shift=(\CF_ifempty{#4}{315}{#4}+\CF_arrowcurrentangle:\CF_ifempty{#5}{5pt}{#5})](line@end@i);
    \draw(line@start)--(line@end);%
    \draw(line@start@i)--(line@end@i);%
    \CF_arrowdisplaylabel{#1}{0.5}+\CF_arrowstartnode{#2}{0.5}-\CF_arrowendnode
}
\catcode`\_=8

\usepackage{pict2e,picture,graphicx}
\usepackage[overload]{empheq}
\usepackage{microtype}
\usepackage{graphicx}
\usepackage{subcaption}
\usepackage{booktabs}
\usepackage{amsmath}
\usepackage{cases}
\usepackage{mathtools}
\usepackage{amsthm}
\usepackage{thm-restate}
\usepackage{enumitem}

\allowdisplaybreaks
\usepackage{xcolor}
\usepackage{nicefrac, xfrac}
\usepackage{xparse}
\usepackage[capitalize, noabbrev]{cleveref}
\usepackage{wrapfig}
\usepackage{bbm}
\usepackage{yfonts}
\usepackage{nicefrac} 
\usepackage{multirow}
\usepackage{bm}
\usepackage{bbm}
\usepackage{subcaption}

\usepackage[T1]{fontenc}

\hypersetup{
	colorlinks = true,
	linkcolor = typ_blue,
	citecolor = Sepia,
	linktocpage = true,
	urlcolor = darkGreen
}

\usepackage{comment}

\usepackage{colortbl}
\usepackage{cellspace}
\usepackage{mleftright}

\usepackage{booktabs} 
\usepackage[ruled]{algorithm2e} 

\crefname{algocf}{alg.}{algs.}
\Crefname{algocf}{Algorithm}{Algorithms}
\crefalias{AlgoLine}{line}
\crefalias{algocfline}{line}
\SetAlFnt{\small}
\SetAlCapFnt{\small}
\SetAlCapNameFnt{\small}
\SetAlCapHSkip{0pt}
\IncMargin{-\parindent}

%% file: settings/defs.tex
\newcommand{\PE}{\textsc{Phased-Elimination}\xspace}

\colorlet{darkgreen}{green!70!black}
\newcommand{\Scal}{\mathcal{S}}
\newcommand{\Bcal}{\mathcal{B}}

\newcommand{\NP}{\textsf{NP}}
\newcommand{\RP}{\textsf{RP}}
\newcommand{\PolyClass}{\textsf{P}}

\newcommand{\supp}{\textnormal{supp}}

\newcommand{\OPT}{\textsf{OPT}}
\newcommand{\diam}{\textnormal{diam}}
\newcommand{\Up}{U^\textsf{P}}
\newcommand{\Ua}{U^\textsf{A}}
\newcommand{\Pcal}{\mathcal{P}}

\newcommand{\Rcal}{\mathcal{R}}
\newcommand{\Ecal}{\mathcal{E}}

\newcommand{\poly}{\textsf{poly}}
\newcommand{\PTAS}{\textsf{PTAS}}
\newcommand{\FPTAS}{\textsf{FPTAS}}

\def\ie{\textit{i.e.}\@\xspace}

\newcommand{\Reals}{\mathbb{R}}

\newcommand{\Naturals}{\mathbb{N}}

\usepackage{color}
\definecolor{mygreen}{rgb}{0.0, 0.5, 0.0}
\definecolor{myorange}{rgb}{0.55, 0.62, 1}
\newcommand{\nb}[3]{{\colorbox{#2}{\bfseries\sffamily\scriptsize\textcolor{white}{#1}}}{\textcolor{#2}{\sf\small\textsf{#3}}}}

\newcommand{\ma}[1]{\nb{Marti}{niceRed}{#1}}

\theoremstyle{plain}

\makeatletter
\let\cref@old@stepcounter\stepcounter
\def\stepcounter#1{%
  \cref@old@stepcounter{#1}%
  \cref@constructprefix{#1}{\cref@result}%
  \@ifundefined{cref@#1@alias}%
    {\def\@tempa{#1}}%
    {\def\@tempa{\csname cref@#1@alias\endcsname}}%
  \protected@edef\cref@currentlabel{%
    [\@tempa][\arabic{#1}][\cref@result]%
    \csname p@#1\endcsname\csname the#1\endcsname}}
\makeatother

%% file: settings/environments.tex
\theoremstyle{plain}
\newtheorem{theorem}{Theorem}[section]

\newtheorem{claim}[theorem]{Claim}
\newtheorem{lemma}[theorem]{Lemma}

\newtheorem*{theorem-non}{Theorem}

\theoremstyle{definition}
\newtheorem{definition}[theorem]{Definition}
\newtheorem{assumption}[theorem]{Assumption}

\theoremstyle{remark}
\newtheorem{remark}[theorem]{Remark}

%% file: settings/colors.tex
\definecolor{niceRed}{RGB}{190,38,38}
\definecolor{Red2}{RGB}{219, 50, 54}
\definecolor{mgreen}{RGB}{160, 200, 140}
\definecolor{blueGrotto}{RGB}{5,157,192}
\definecolor{limeGreen}{HTML}{81B622}
\definecolor{myellow}{rgb}{0.88,0.61,0.14}
\definecolor{darkGreen}{HTML}{2E8B57}
\definecolor{navyBlueP}{HTML}{03468F}
\definecolor{Sepia}{HTML}{7F462C}
\definecolor{red2}{HTML}{1F462C}
\definecolor{orange2}{HTML}{FF8000}
\definecolor{mgray}{HTML}{ABB3B8}
\definecolor{lgray}{HTML}{E5E8E9}
\definecolor{myPurple}{RGB}{175,0,124}
\definecolor{mypurple2}{rgb}{0.8,0.62,1}
\definecolor{royalBlue}{HTML}{057DCD}
\definecolor{mpink}{HTML}{FC6C85}
\definecolor{lblue}{RGB}{74,144,226}
\definecolor{peagreen}{RGB}{152,193,39}
\definecolor{typ_navy}{HTML}{001f3f}
\definecolor{typ_blue}{HTML}{0074d9}
\definecolor{typ_aqua}{HTML}{7fdbff}
\definecolor{typ_teal}{HTML}{39cccc}
\definecolor{typ_eastern}{HTML}{239dad}
\definecolor{typ_purple}{HTML}{b10dc9}
\definecolor{typ_fuchsia}{HTML}{f012be}
\definecolor{typ_maroon}{HTML}{85144b}
\definecolor{typ_red}{HTML}{ff4136}
\definecolor{typ_orange}{HTML}{ff851b}
\definecolor{typ_yellow}{HTML}{ffdc00}
\definecolor{typ_olive}{HTML}{3d9970}
\definecolor{typ_green}{HTML}{2ecc40}
\definecolor{typ_lime}{HTML}{01ff70}
\definecolor{newgreen}{HTML}{83c702}
\definecolor{newpurp}{RGB}{97,96,121}

%% file: src/abstract.tex
\begin{abstract}

We study a Bayesian contract design problem in which a principal interacts with an unknown agent.
We consider the single-parameter uncertainty model introduced by \citet{alon2021contracts}, in which the agent's type is described by a single parameter, \emph{i.e.}, the cost per unit-of-effort.
Despite its simplicity, several works have shown that single-dimensional contract design is not necessarily easier than its multi-dimensional counterpart in many respects. Perhaps the most surprising result is the reduction by \citet{castiglioni2025reduction} from multi- to single-dimensional contract design. However, their reduction preserves only multiplicative approximations, leaving open the question of whether additive approximations are easier to obtain than multiplicative ones.

%
%
%

In this paper, we give a positive answer---at least to some extent---to this question.
In particular, we provide an additive $\PTAS$ for these problems while also ruling out the existence of an additive $\FPTAS$. This, in turn, implies that no reduction from multi- to single-dimensional contracts can preserve additive approximations. Moreover, we show that single-dimensional contract design is fundamentally easier than its multi-dimensional counterpart from a learning perspective. Under mild assumptions, we show that optimal contracts can be learned efficiently, providing results on both regret and sample complexity.
\end{abstract}

%% file: src/introduction.tex
\section{Introduction}

Contract design studies the interaction between a principal and an agent.
The agent undertakes costly actions that are hidden from the principal, who can only observe a stochastic outcome influenced by the agent's decisions.
The goal of the principal is to design an outcome-dependent payment scheme, known as a \emph{contract}, to incentivize the agent to take desirable actions.

When all the parameters of the model are known, an optimal contract can be computed efficiently through linear programming~\citep{dutting2019simple}.
However, this changes dramatically when uncertainty is added to the model. When Bayesian uncertainty is introduced and an agent's type, defining their actions and costs, is sampled from a probability distribution, contract design problems become intractable~\citep{castiglioni2022bayesian,guruganesh2021contracts,castiglioni2023designing}.
%
A comparable increase in the complexity of the problem occurs in online learning scenarios, where the principal interacts with an unknown agent. In this setting, uncertainty makes the problem intractable, leading to regret bounds that are nearly linear~\citep{zhu2022online}, or exponential in the instance size~\citep{bacchiocchi2023learning}.

These negative results motivate the study of simpler settings with more ``structured'' uncertainty. Drawing inspiration from mechanism design, \citet{alon2021contracts} introduced a setting in which the agent's type is single-dimensional and represents a cost per unit-of-effort, \emph{i.e.}, a parameter that multiplies the cost of the chosen action.
A recent line of work~\citep{alon2021contracts,alon2023bayesian,castiglioni2025reduction} shows that, while single-dimensional contract design enjoys some additional structural properties, it is computationally more similar to the multi-dimensional case than initially expected.
These results are culminated in \citet{castiglioni2025reduction}, who provided an almost approximation preserving reduction from multi-dimensional to single-dimensional contract design for \emph{multiplicative} approximations. 

At first glance, the above result may suggest that efficient algorithms for contract design with single-dimensional costs are unattainable.
However, we show that, in several key aspects, single-dimensional contract design is simpler than its multi-dimensional counterpart, and achieving better results in this setting is possible.
In particular, we answer positively to the following questions:
\begin{enumerate} 
    \item Is it possible to design efficient approximation algorithms for Bayesian single-dimensional contract design under additive approximations?
    \item Is it possible to design learning algorithms with sublinear regret (or with polynomial sample complexity) in online learning scenarios with single-dimensional types?
\end{enumerate}

\subsection{Our Contributions}

%
In the first part of the paper, we analyze the single-dimensional Bayesian contract design problem under \emph{additive} approximations. While single-dimensional and multi-dimensional problems are essentially equivalent under multiplicative approximations \cite{castiglioni2025reduction}, this is not the case for additive ones.
Indeed, we show that single-dimensional problems admit a $\PTAS$. This is in contrast to multi-dimensional problems, which do not admit a $\PTAS$ unless $\PolyClass=\NP$~\cite{guruganesh2021contracts, castiglioni2022bayesian}. 

\begin{theorem-non}[Informal, see \Cref{thm:PTAS}]
    Single-dimensional Bayesian contract design admits an additive $\PTAS$.
\end{theorem-non}

Our result relies on the following observation.
In multi-dimensional instances, all types might have very different actions, costs, and incentives. This is not the case in single-dimensional instances, where similar types share similar incentives.
Our algorithm exploits this feature by grouping agents with similar types into $k$ buckets.
Then, replacing each bucket with a representative type, we can employ algorithms for discrete types such as the one proposed by ~\citet{guruganesh2021contracts} and \citet{castiglioni2022bayesian} to solve the discretized problem. Such algorithms run in time exponential in number of types. However, if the number of representative types (and thus the number of buckets $k$) is constant, they run in polynomial time.

The core of the proof is to manage the error with respect to the optimal contract as a function of the number of buckets $k$ in which we partition the types.
The agent's type is represented by a single parameter $\theta\in [0,1]$. 
We will show that when we group two agents of types $\theta \in [0,1]$ and $\theta'\in [0,1]$ together, we make an error of order $\poly(|\theta-\theta'|)$. Thus, to guarantee an $\epsilon=O(1)$ additive error, it suffices to set $k=\poly(1/\epsilon)$, resulting in a $\PTAS$. 
To do that, we relate the value of the discretized problem to the true one. Then, we make the solution of the discretized problem robust to the slight type misspecification resulting from bucketing different types together.
%
%
This result shows that focusing on multiplicative approximations is necessary to design an approximation-preserving reduction from multi- to single-dimensional contracts~\citep{castiglioni2025reduction}. Indeed, for additive approximations, single-dimensional contracts admit a $\PTAS$ while multi-dimensional ones are $\NP$-hard \citep{guruganesh2021contracts, castiglioni2022bayesian}. 

It is therefore natural to try to strengthen our result and ask whether an $\FPTAS$ exists for single-dimensional contracts with additive approximations. Unfortunately, we answer this question in the negative.
Indeed, we show that our $\PTAS$ is tight, and it is impossible to design an additive $\FPTAS$ unless $\PolyClass=\NP$.

\begin{theorem-non}[See \Cref{thm:reduction}]
    Single-dimensional Bayesian contract design does not admit an additive $\FPTAS$, unless $\PolyClass=\NP$.
\end{theorem-non}

The closest related work to this hardness result is~ \citet{castiglioni2025reduction}, which establishes the equivalence of multi-dimensional and single-dimensional contract design under multiplicative approximations. A key aspect of their approach is using exponentially decreasing costs and expected rewards, leading to an exponentially small principal's utility as the size of the instance grows. In such cases, the null contract incurs only a negligible additive error. 
In contrast, we focus on additive approximations and construct instances with polynomially decreasing costs and expected rewards, ensuring that the principal’s utility remains polynomially bounded. 
Exponentially decreasing sequences are used to ``separate'' the actions of each type from the actions of other types by using incentives. However, using polynomial sequences results in sets of actions that are less separated. Dealing with these less separated action sets is the primary challenge of the proof.

In the second part of the paper, we focus on the problem of learning approximately optimal contracts, providing bounds both on the regret rates and on the sample complexity.
We study a stochastic setting for which we make some mild assumptions. In particular, as in previous work on online learning in contract design \citep{zhu2022online}, we restrict to bounded contracts. Moreover, as in many previous works on single-dimensional contract design, we assume that the type distribution has bounded density~\citep{alon2021contracts,alon2023bayesian}.

There are two main challenges in obtaining no-regret algorithms in our setting.
First, the action space is continuous (being the hypercube) and lacks any obvious structural properties to exploit. This makes most techniques used in multi-dimensional contract design unfeasible, as they operate directly on this decision space and cannot leverage the simplicity of single-dimensional types~\cite{zhu2022online,bacchiocchi2023learning,ho2014adaptive}.
Second, the type distribution is continuous.
This makes standard approaches for Stackelberg-like games unfeasible. Indeed, most approaches for learning in games with commitments crucially rely on having a finite number of types  (see \emph{e.g.}, \citet{balcan2015commitment,castiglioni2020online}).

We address such challenges by considering an approximate instance with finitely many types, whose grid of discretized types is the one used in the Bayesian problem.
We show that, under the bounded density assumption, the utility of each contract does not change too much from the original instance to the discretized one, allowing us to work with discretized types.
This result is conceptually very different from the ``continuous-to-discrete'' result we establish when working towards the $\PTAS$. In particular, we now provide both upper and lower bounds on utility, but only for distributions with bounded density.
Then, similarly to \citet{bernasconi2023optimal} we shift the learning problem to the space of utilities. 
This approach allows us to avoid the exponentially large set of actions that would typically need to be considered (where bandit feedback would translate the exponential number of actions into a regret bound that grows exponentially with the instance size).
Indeed, in this context, the problem is an \emph{almost-}linear bandit whose dimension is equal to the number of types. It is important to note that we sacrifice the linear structure of the problem, as the discretized instance is only an approximation of the original one.
We deal with this almost-linear structure by resorting to \emph{misspecified linear bandits}~\citep{ghosh2017misspecified}, and by carefully balancing the effects of the discretization on the misspecification of the linear model.
Combining all these results, we obtain a regret minimization algorithm that guarantees regret of order $\widetilde O(T^{3/4})$.

\begin{theorem-non}[Informal, see \Cref{thm:regret}]
There exists an algorithm that guarantees regret
\(
R_T= \widetilde O(\poly(\mathcal{I})T^{3/4}),
\)
where $\mathcal{I}$ is the instance size.
\end{theorem-non}

Finally, we show how the same reduction to misspecified linear bandits can be exploited to construct an algorithm with polynomial sample complexity.

\begin{theorem-non}[Informal, see \Cref{thm:sample}]
    Given an $\eta>0$, there exists an algorithm that with high probability finds an $\eta$-optimal contract using $\widetilde O\left(\frac{\poly(\mathcal{I})}{\eta^4}\right)$ samples,
    where $\mathcal{I}$ is the instance size.
\end{theorem-non}

\subsection{Related Works}

In this section, we review the studies most relevant to our work. For a broader overview on algorithmic contract theory, we refer readers to the recent survey by \citet{dutting2024algorithmic}. 

The computational study of contract design has attracted growing interest in recent years, beginning with \citet{babaioff2006combinatorial}, who introduced the multi-agent setting, and \citet{dutting2021complexity}, who explored computational aspects of the single-agent model. 
A more recent line of work extends contract analysis to Bayesian settings, where the principal has a known probability distribution over the agent’s types~\citep{guruganesh2021contracts, castiglioni2021bayesian, alon2021contracts, alon2023bayesian}. In particular, \citet{castiglioni2023designing} introduce menus of randomized contracts and show that they can be computed efficiently. \citet{gan2022optimal} generalize the principal-agent problem to cases where the principal’s action space is potentially infinite and subject to design constraints, proving that optimal randomized mechanisms remain efficiently computable.  


\paragraph{Single-dimensional Contract design}
The computational perspective on single-parameter models was first introduced by \citet{alon2021contracts,alon2023bayesian}.
The study of the single parameter setting is mainly motivated by the known computational complexity separation between single-parameter and multi-parameter settings in algorithmic mechanism design \cite{nisan1999algorithmic,daskalakis2015multi}. However, \citet{castiglioni2025reduction} showed that, in the context of contract design, the relationship between the two settings is more intricate. In particular, they provided an almost approximation-preserving reduction from multi- to single-dimensional contract design for \emph{multiplicative} approximations.

\paragraph{Learning Optimal Contracts} The problem of learning optimal contracts was first explored by \citet{ho2014adaptive}, who focused on instances with a very specific structure, similar to what is done in \citet{cohen2022learning} (see the discussion in \citet{zhu2022online} for more details about these assumptions). \citet{zhu2022online} extended this line of work by studying the problem in general principal-agent settings. They proposed an algorithm with a regret upper bound of $\widetilde O(m^{1/2} T^{1-1/(2m+1)})$ and established an almost matching lower bound of $\Omega(T^{1-1/(m+2})$. \citet{bacchiocchi2023learning} developed an algorithm that achieves cumulative regret with respect to an optimal contract upper bounded by $\widetilde O(m^n T^{4/5})$, which remains polynomial in the instance size when the number of agent actions $n$ is constant.
Our problem also shares connections with the more general problem of learning an optimal commitment, in repeated Stackelberg games \cite{letchford2009learning,peng2019learning,blum2014learning,balcan2015commitment} and online Bayesian persuasion \cite{castiglioni2020online,castiglioni2021multi,bernasconi2023optimal}.




%% file: src/preliminaries.tex
\section{Preliminaries}

\subsection{The Bayesian Principal-Agent Problem}\label{sec:preliminaries_problem}

An instance of a principal-agent problem is defined by a set $A$ of $n \coloneqq |A|$ agent's (hidden) actions. 
Each action leads stochastically to one of the $m$ possible different outcomes $\omega\in \Omega$, which provides  a reward $r_\omega$ to the principal. Each action defines a different probability distribution on the outcomes. We denote with $F_{a,\omega}$ the probability that action $a\in A$ leads to outcome $\omega\in \Omega$. Thus, we have that $F_{a,\omega}\in[0,1]$ for each $a\in A$ and $\omega\in\Omega$, and $\sum_{\omega\in\Omega}F_{a,\omega}=1$ for each $a\in A$. We consider the \emph{single-dimensional hidden-type model} introduced by \citet{alon2021contracts}, in which each action $a$ requires $c_a$ units of effort from the agent. Each agent has a hidden type $\theta\in\Theta\coloneqq [0,1]$, which represents their cost per unit of effort. Consequently, an agent of type $\theta$ incurs a cost of $c_{a,\theta}=\theta\cdot c_a$ when performing action $a\in A$.
We say that this setting is single-dimensional in the sense that the agent's type $\theta$ is a single real value that only influences the cost of the agent's actions. In contrast, in general \emph{multi-dimensional} contract design problems, both the action probability $F_{a,\theta}$ and the action cost $c_{a,\theta}$ can depend arbitrarily on the agent's type $\theta$.

The goal of the principal is to design a contract which specifies a payment for each possible outcome. We denote with $p\in\Reals^m_{\ge 0}$ such a payment vector where $p_\omega$ defines the payment when outcome $\omega$ realized.\footnote{The assumption that payments are non-negative is a standard assumption in contract design, known as \emph{limited-liability}~\citep{carroll2015robustness, innes1990limited}.}

When an agent of type $\theta$ takes action $a$, then their expected utility is
\(
\Ua_\theta(p,a):=\sum_{\omega \in \Omega} F_{ a, \omega } p_\omega - \theta c_{a}.
\)
On the other hand, the principal expected utility is 
\(
\Up(p,a):= \sum_{\omega \in \Omega} F_{a, \omega} (r_\omega - \, p_\omega).
\)
Given a contract $p\in\Reals^m_{\ge 0}$ an agent of type $\theta$ plays an action that is
\begin{enumerate}
	\item \emph{incentive compatible} (IC), \emph{i.e.}, it maximizes their expected utility over actions in $A$; and
	\item \emph{individually rational} (IR), \emph{i.e.}, it has non-negative expected utility.
\end{enumerate}

To simplify exposition we can assume w.l.o.g.~that there is an action $a\in A$ such that $c_a=0$, thus conveniently including IR into the IC constraints. In the case in which there are multiple actions which satisfy the IC constraints, we assume (as it it standard in contract design \citep{carroll2015robustness}) that ties are broken in favor of the principal. We can define the set $\Bcal^\theta(p)$ of all actions which are best-responses for an agent of type $\theta$ to a contract $p$. For ease of presentation, we will also need to define the set of actions $\Bcal_\epsilon^\theta(p)$ which are $\epsilon$-IC, formally
\[
{\Bcal}^{\theta}_\epsilon(p)\coloneqq\left\{a\in A: \Ua_\theta(p,a) \ge \Ua_\theta(p,a')-\epsilon\quad\forall a'\in A\right\},
\]
and clearly ${\Bcal}^{\theta}(p)={\Bcal}^{\theta}_0(p)$. The assumption that ties are broken in favor of the principal can be written by assuming that the action $b^\theta(p)\in A$ taken by an agent of type $\theta$ for a contract $p$, is any action such that
$b^\theta(p)\in \arg \max_{a \in \Bcal^{\theta}(p)} \Up(a,p).$ \footnote{It is convenient, especially in \Cref{sec:learning}, to fix the way in which ties are further broken, and we assume that $b^\theta(p)$ is chosen in the set $\arg \max_{a \in \Bcal^{\theta}(p)} \Up(a,p)$ according to a fixed ordering of the agent's actions.}

\subsection{Single-Dimensional Bayesian Principal-Agent Problem}

In the first part of the paper, we focus on \emph{single-dimensional Bayesian} principal-agent problems, in which the type of each agent is drawn from a publicly-known distribution $\Gamma$ over types. This distribution can be absolutely continuous, discrete, or a combination of the two.

With slight abuse of notation, we write $\Up(p,\theta)$ instead of $\Up(p,b^\theta(p))$ to denote the utility of the principal when the agent's realized type is $\theta$.
Then, given an instance of our problem, we define $\OPT\coloneqq\sup_{p\in\Reals_{\ge 0}^m}\Up(p)$, where $\Up(p)$ is the principal expected utility of contract $p$ which can be written as
\[
\Up(p)\coloneqq\mathbb{E}_{\theta \sim \Gamma} \left[\Up(p,\theta)\right].
\]

\subsection{Learning an Optimal Contract}\label{sec:prelimLearning}

In the second part of the paper, we focus on a regret minimization setting in which the principal interacts sequentially with a stream of $T$ agents.
At each round $t\in [T]$, the type $\theta_t$ of the $t$-th agent is sampled accordingly to the distribution $\Gamma$, which is unknown to the principal.
%
Then, the principal commits to a contract $p_t\in\Reals^m_{\ge 0}$, and the agent plays the best response $a_t=b^{\theta_t}(p_t)$. Finally, the principal observes the outcome $\omega \sim F_{a_t}$, and receives utility $r_{\omega_t}-p_{\omega_t}$. 

The principal's regret is defined as:
\[
R_T= T \cdot\OPT- \mathbb{E}\left[\sum_{t \in T}  \mathbb{E}_{\theta\sim \Gamma}[\Up(p_t,\theta)]\right],   
\]
where the expectation is on the randomness of the algorithm.
Our goal is to design algorithms with regret sub-linear in $T$ and polynomial in the instance size.

Another problem studied in the learning literature consists in finding an approximately optimal contract with high probability. The interaction between the learner and the environment is the same as in the regret-minimizing scenario. However, the goal is to minimize the number of rounds needed to find a good approximation of the optimal contract.
Formally, given an approximation error $\eta$ and a probability $\delta$, we aim at finding a contract $p$ such that
\[
\mathbb{P}\left[\mathbb{E}_{\theta\sim \Gamma}\left[\Up(p,\theta)\right]\ge \OPT-\eta\right]\ge 1-\delta.
\]
The goal is to find such a contract with a small number of samples $T$.

%% file: src/additivePTAS2.tex
\section{An Additive $\PTAS$ for Single-Dimensional Costs and Continuous-Types}\label{sec:PTAS}

In this section, we show that the problem admits an additive $\PTAS$. 
This is in contrast to the multi-dimensional problem, that does not admits an additive $\PTAS$, unless $\PolyClass=\NP$ (see \cite{guruganesh2021contracts} and \cite[Theorem~7 for $\rho=1$,]{castiglioni2022bayesian}). 

While in instances with multi-dimensional costs, all types might have very different incentives, our result relies on the simple observation that in single-dimensional instances, similar types share similar incentives and behavior.
Our algorithm exploits this feature and groups agents with similar types into $k$ buckets.
Then, replacing each bucket with a representative type, we can employ algorithms for discrete types such as the one proposed by ~\citet{guruganesh2021contracts} and \citet{castiglioni2022bayesian} to solve the discretized problem. Such algorithms run in time $\poly(n^k, m)$ and therefore in polynomial time when the number of buckets $k$ is constant.
The core of the proof is to manage the error with respect to $\OPT$ in terms of the number of buckets $k$. In particular, we will show that when we group two agents of type $\theta$ and $\theta'$ together, we make an error of order $\poly(|\theta-\theta'|)$. Thus, to guarantee an $\epsilon=O(1)$ additive error, we can set $k=\poly(1/\epsilon)$, resulting in a $\PTAS$.

More specifically, our algorithm works as follows. 
Given an instance of single-dimensional Bayesian contract design (possibly with continuous types), we build an instance of contract design with discretized types as follows.
Given a parameter $\delta$ to be set in the following, let $\theta_i=(i-1/2) \delta$ for each $i \in \lceil 1/\delta\rceil$.
Then, we consider the set of types $\Theta_{\delta}:=\{\theta_i\}_{i\in\lceil 1/\delta\rceil}$ each realizing with probability 
\begin{align*}
    \gamma_{\theta_i}\coloneqq \begin{cases}
    \mathbb{P}_{\theta\sim\Gamma}[\theta\in{(\theta_i-\delta/2,\theta_i+\delta/2]}]&\text{if $i>1$}\\
    \mathbb{P}_{\theta\sim\Gamma}[\theta\in{[0,\delta/2]}]&\text{if $i=1$}
    \end{cases}.
\end{align*}
Here, we do not require any particular assumption on how the distribution is represented. Computationally, we only make the mild assumption to have access to an oracle, returning the probability that $\theta\sim \Gamma$ belongs to a given interval.\footnote{It will be clear from our proof that our algorithm is robust to approximations in the probabilities returned by the oracle.}
As we already observed, this discretized problem with types supported on $\Theta_\delta$ can be solved in polynomial time for any constant $\delta>0$~\citep{guruganesh2021contracts, castiglioni2022bayesian}.
Let $\tilde p$ be the contract computed on the discretized problem.
Then, our algorithm returns as a solution the contract $p= \tilde p +\alpha (r-\tilde p)$, where $\alpha$ will be set in the following. Intuitively, this step makes the contract robust to the use of approximate types. 

Formally, we obtain the following guarantees. 
\begin{theorem}\label{thm:PTAS}
    There is an algorithm that for any $\epsilon>0$ runs in time $\poly(n^{{1}/{\epsilon}},m)$ and finds a contract $p\in\Reals^m_{\ge 0}$ such that $\Up(p)\ge \OPT-O(\epsilon)$.
\end{theorem}

Technically, our result is based on two main observations:
\begin{itemize}
\item If the payment vector $p= \tilde p +\alpha(r-\tilde p)$ is applied to the original instance of the problem, it yields a principal's utility close to the utility of contract $\tilde p$ applied to the discretized instance;
\item We relate the utility achieved by an optimal contract in the original and the discretized instance, showing that the latter is only slightly smaller than the former.
\end{itemize}
The theorem follows from the combination of these two results.

We start the formal analysis providing some thecnical lemmas. First, we notice that an optimal contract for the discretized instance that incentivizes action $a_i$ for a type $\theta_i\in \Theta_\delta$, approximately incentivizes the same action for types close to $\theta_i$ in the original instance.
As a tool to exploit this $\delta$-best response, we will use the well-known result of \citet{dutting2021complexity} that relates IC and approximate-IC contracts. We recall the statement of this result in the following lemma.
\begin{lemma}[{\citep[Proposition~2.4]{dutting2021complexity}}]\label{lem:linearization}
Let $\rho:\Theta\to A$ be a generic mapping from types to actions, specifying the action $\rho(\theta)$ taken by each type $\theta\in\Theta$, and let $p\in\Reals^m_{\ge 0}$ satisfy $\rho(\theta)\in \Bcal^{\theta}_{\epsilon}(p)$ for all $\theta\in\Theta$. Then, $\tilde p:=p+\alpha(r-p)$ guarantees $\Up(\tilde p,\theta)\ge \Up(p,\rho(\theta))-\left(\frac{\epsilon}{\alpha}+\alpha\right)$.
\end{lemma}

Next, we must establish a connection between the discretized and continuous instances of the problem. 
We do that through the following lemma, which has two main implications. 
First, it implies that if for each type $\theta_i\in \Theta_\delta$ the  action played (which may not be IC) does not change across the interval $(\theta_i-\delta/2, \theta_i+\delta/2]$, then the principal's utility in the original and discretized instances coincide.
Second, it shows that the best response of a type $\theta_i \in \Theta_\delta$ is still an approximate best response for each type in $(\theta_i-\delta/2, \theta_i+\delta/2]$.
In the following lemma, we prove a slightly more general result for arbitrary partition of $\Theta=[0,1]$.

\begin{restatable}{lemma}{lemmadiscretization}\label{lem:discretization}
    Let $\Pcal=\{\Pcal_1,\ldots,\Pcal_k\}$ be any finite partition of $[0,1]$. Moreover let $\hat\Theta:=\{\theta_1,\ldots,\theta_k\}$ be such that $\theta_i\in\Pcal_i$ for all $i\in[k]$, and let $\rho:\Theta\to A$ be such that $\rho(\theta)=\rho(\theta_i)$ for each $\Pcal_i$ and $\theta\in\Pcal_i$. Then, for any contract $p\in\Reals^m_{\ge 0}$ it holds
    \[
        \mathbb{E}_{\theta\sim\Gamma} \left[\Up(p, \rho(\theta))\right] = \sum_{\theta_i\in\hat\Theta} \gamma_{\theta_i} \Up(p, \rho(\theta_i)),
    \]
    where $\gamma_{\theta_i}\coloneqq\mathbb{P}_{\theta\sim\Gamma}[\theta\in \Pcal_i]$.
    Moreover, if for any $\theta\in\Pcal_i$, $i \in [k]$, we have $\rho(\theta)\in \Bcal^{\tilde\theta}(p)$ for some $\tilde \theta\in \Pcal_{i}$, then $\rho(\theta)\in\Bcal^{\theta}_{\diam(\Pcal_i)}(p)$ for all $\theta\in\Pcal_i$, where $\diam(\Pcal_i)\coloneqq\sup_{\theta,\theta'\in\Pcal_i}|\theta-\theta'|$.
\end{restatable}

For instance, we can consider the set $\Theta_\delta$ defined above together with the natural partition of $[0,1]$ given by
\(\{[0,\delta],(\delta,2\delta],\ldots(\lfloor\frac1\delta\rfloor\delta,1]\}\).
The first part of \Cref{lem:discretization} states that if $\rho:\Theta\to A$ is a constant function on each partition $\Pcal_i$, then the values under the original and discretized instance are the same. The second part of the statement states that if the function $\rho$ recommends to an agent of type $\theta\in\Pcal_i$ an action that is a best response to a type $\tilde\theta\in\Pcal_i$, then $\rho(\theta)$ is $\diam(\Pcal_i)$-IC for any agent with type belonging to $\Pcal_i$.


Equipped with the lemma above, we can establish the existence of a $\PTAS$.

\begin{proof}[Proof of \Cref{thm:PTAS}]
    Let $\tilde p$ be the optimal contract of the discrete-type instance defined over types $\Theta_\delta$ and with distribution $\{\gamma_\theta\}_{\theta \in \Theta_\delta}$, and let $p\coloneqq \tilde p+\alpha(r-\tilde p)$, where $\delta>0$ and $\alpha>$ will be defined in the following.
    By using the algorithm of \citet[Theorem~8]{castiglioni2022bayesian} or \citet[Lemma~2]{guruganesh2021contracts}, we can build $\tilde p$ in time $\poly(n^{|\Theta_\delta|}, m)$, where $|\Theta_\delta|=\lfloor\frac{1}{\delta}\rfloor$.

    First, we relate the principal's utility when playing $p$ in the original instance to their utility when playing $\tilde p$ in the discretized instance.
    Let $\rho_{\tilde p}^\delta(\theta)$ be the function that assigns to each type $\theta\in [0,1]$ the best response played by the closest type in $\Theta_\delta$. Formally, for all $\theta\in[0,1]$, taking  $\theta_i\in \Theta_\delta$ such that $\theta \in(\theta_i-\delta/2,\theta_i+\delta/2]$, we have that $\rho_{\tilde p}^\delta(\theta) = b^{\theta_i}(\tilde p)$.
    Then, by the second part of \Cref{lem:discretization}, for each $\theta \in [0,1]$ and  $\theta_i \in \Theta_\delta$  such that $\theta \in(\theta_i-\delta/2,\theta_i+\delta/2]$, it holds
    \[ \rho_{\tilde p}^\delta(\theta)\in \Bcal^{\theta_i}(\tilde p) \subseteq \Bcal_{\delta}^{\theta}(\tilde p).   \]
    %
    Thus, by \Cref{lem:linearization}, we can readily say that for each $\theta\in\Theta$ it holds
	\begin{align*}
	\Up(p,\theta)+\left(\frac{\delta}{\alpha}+\alpha\right)\ge \Up(\tilde p, \rho^\delta_{\tilde p}(\theta)).
	\end{align*}
	By taking expectations over $\Gamma$ we find that
	\begin{align}\label{eq:ptastmp1}
	\Up(p)+\left(\frac{\delta}{\alpha}+\alpha\right)\ge \mathbb{E}_{\theta\sim\Gamma} \left[\Up(\tilde p, \rho^\delta_{\tilde p}(\theta))\right].
	\end{align}

    Now, we have to relate the optimal value of the discretized instance to the optimal value of the original one, showing that these are not too dissimilar.
    
    Instrumentally to this goal, we can apply the first part of \Cref{lem:discretization} to the right-hand side of \Cref{eq:ptastmp1}, obtaining:
	\begin{align}\label{eq:ptasAdded}
	\mathbb{E}_{\theta\sim\Gamma} \left[\Up(\tilde p, \rho^\delta_{\tilde p}(\theta))\right]= \sum_{\theta_i\in\Theta_{\delta}} \gamma_{\theta_i} \Up(\tilde p, \theta_i).
	\end{align}
	Now notice that $\tilde p$ is the optimum with respect to the discretized instance and thus  
    \begin{align}\label{eq:ptasAdded2}
    \sum_{\theta_i\in\Theta_{\delta}} \gamma_{\theta_i} \Up(\tilde p, \theta_i)\ge \sum_{\theta_i\in\Theta_{\delta}} \gamma_{\theta_i} \Up(\bar p, \theta_i)\quad \forall \bar p\in \Reals^m_{\ge 0}.
    \end{align}
	In particular, we can take $\bar p=p^{\star}_k+\alpha(r-p^{\star}_k)$, where $p^{\star}_k$ is any contract such that $\Up(p^{\star}_k)\ge \OPT-\frac{1}{k}$ where $k\in\Naturals$ (we observe that a contract that attains the supremum might not exist).

        Now, for $i\in \lfloor \frac{1}{\delta}\rfloor$ and $k\in\Naturals$, we define the set of possible best responses of types in the interval $\Pcal_i:=(\theta_i-\delta/2,\theta_i+\delta/2]$ for contract $p^\star_k$ as 
        \[
        \mathcal{A}_{k,i}=\bigcup_{\theta\in\Pcal_i} \Bcal^{\theta}(p^\star_k).
        \]
        Then, we define a function $\rho:\Theta\to A$ so that, for any $\theta\in\Pcal_i$, it holds
        \[
	\rho(\theta) = \arg\max_{a\in\mathcal{A}_{k,i}} \Up( p^\star_k,a).
	\]
   Intuitively, $\rho(\theta)$ is the action that maximizes the principal's utility among all actions that are IC for a type $\theta'$ that is in the same partition as $\theta$.
    Now, thanks to the second part of \Cref{lem:discretization}, $\rho(\theta)\in\Bcal_\delta^{\theta}(p^\star_k)$.
    Hence, we can apply \Cref{lem:linearization} and take the expectation over $\Gamma$ obtaining
	\begin{align}\label{eq:ptastmp2}
    \sum_{\theta_i\in\Theta_{\delta}}\gamma_{ \theta_i}\Up(\bar p, \theta_i)\ge \sum_{\theta_i\in\Theta_{\delta}}\gamma_{\theta_i}\Up(p_k^\star, \rho(\theta_i))-\left(\frac\delta\alpha+\alpha\right).
	\end{align}
	Since $\rho(\theta)$ is constant on each partition, we can use the first part of \Cref{lem:discretization} to get that
	\[
	\sum_{\theta_i\in\Theta_{\delta} }\gamma_{\theta_i}\Up(p_k^\star, \rho(\theta_i))=\mathbb{E}_{\theta\sim\Gamma} \left[\Up(p_k^\star, \rho(\theta))\right].
	\]
	Now, consider each term $\Up(p_k^\star, \rho(\theta))$ separately. By our construction of $\rho(\theta)$ it is easy to see that 
    \begin{align}\label{eq:PTASFinal}
    \Up(p^\star_k, \rho(\theta))\ge \Up(p_k^\star, \theta).
    \end{align}
    Combining, \Cref{eq:ptastmp1,eq:ptastmp2,eq:ptasAdded,eq:ptasAdded2,eq:PTASFinal} we get that \(\Up(p)\ge\mathbb{E}_{\theta\sim\Gamma}[\Up(p^\star_k,\theta)]-2\left(\frac\delta\alpha+\alpha\right)\). Therefore,
	\begin{align*}
	\Up(p)&
    \ge \mathbb{E}_{\theta\sim\Gamma}[\Up(p^\star_k,\theta)]-2\left(\frac\delta\alpha+\alpha\right) \\
    &=\Up(p^\star_k)-2\left(\frac\delta\alpha+\alpha\right)\ge \OPT-\frac{1}{k}-2\left(\frac\delta\alpha+\alpha\right).
	\end{align*}
    Choosing $\alpha=\sqrt{\delta}=\epsilon^2/16$ and taking the limit $k\to+\infty$ concludes the proof.	
\end{proof}

%% file: src/hardness.tex
\section{Impossibility of an Additive $\FPTAS$}\label{sec:hardness}

In this section, we show that the additive $\PTAS$ obtained in the previous section is tight. In particular, we show that the single-dimensional Bayesian principal-agent problem does not admit an additive $\FPTAS$, even if the distribution has finite support. As it is customary for distributions with finite support, we assume that the distribution $\Gamma$ is represented by a tuple $(\gamma_\theta)_{\theta \in \supp(\Gamma)}$.

\begin{theorem}\label{thm:reduction}
The single-dimensional Bayesian principal-agent problem does not admit an additive $\FPTAS$, unless $\PolyClass=\NP$.
\end{theorem}

\paragraph{Challenges in Hardness of Single-Dimensional Contract Design}

When proving the hardness of multi-dimensional contract design, a recurrent feature that is exploited is the fact that an agent of type $\theta$ essentially has different parameters. For example, we could easily implement an instance in which the set of actions is different for each agent.\footnote{
Not all existing reductions adopt this framing. For example, in the reduction by \citet{guruganesh2021contracts}, this is achieved by assigning a distinct outcome distribution to each type $\theta$. However, this is essentially equivalent to having a different set of actions, as presented in our discussion.}
Next, we provide a high-level intuition on how this can be leveraged into a reduction. For example, if reducing from \textsc{Set-Cover}, one can have a type for each element of the universe $E$ and an outcome for each set $S$ in the universe $\Scal$. Then, a type corresponding to an element in the universe can only play actions that lead deterministically to a ``favorable outcome'' $\omega_S$. For any cover $\Scal^\star\subseteq\Scal$, we can thus construct a contract that incentivizes an action that leads to $\omega_S$ if $S\in\Scal^\star$. Since $\Scal^\star$ is a cover of $E$, each agent has an action that leads to a favorable outcome, and thus, the utility of the player is large. The fundamental property of this reduction is that an agent of a type corresponding to an element $i$ can only induce outcomes $\omega_S$ such that $i\in S$. However, this is not the case in single-dimensional settings.
Indeed, the main technical challenge in proving hardness in the single-dimensional setting is that all agents have the same action set, outcomes, and outcome distributions. They can only differ in their costs, and only by a multiplicative factor! This constraint makes proving hardness results challenging in single-dimensional settings, as it requires artificially ``separating'' actions.

\paragraph{Previous Approach} The closest construction to ours is presented in~\citet{castiglioni2025reduction}, which reduces multi-dimensional to single-dimensional contract design, preserving \emph{multiplicative} approximations. 
One crucial feature of their approach is the use of sequences of exponentially decreasing costs and expected rewards.
Sequences of exponentially decreasing costs and expected rewards are very effective in creating ``separation'' among actions, making some actions clearly suboptimal for an agent of a given type.
At a high level, their approach essentially partitions the set of actions, guaranteeing that only inducing the agent to play a subset of action (the ones available to the type in the multi-dimensional instance) provides a non-negligible reward to the principal.
While this is effective for multiplicative approximations, it results in the principal’s expected utility becoming exponentially small as the instance size increases. In such instances, approximately optimal contracts can be efficiently found under additive approximation. Indeed, posting the null contract provides an additive error that is trivially exponentially small in the instance size.

\paragraph{High-Level Construction} For the reason above, we recognized the need to use polynomially decreasing sequences of costs and expected rewards, which require a much more fine-grained analysis to establish such separation. Here, we outline the construction of our sequences, temporarily setting aside the additional complexities of the reduction.
%
The goal is to find a sequence of outcome distributions $F_{a_i,\omega}$,  action costs $c_{a_i}$, and a sequence of types $\theta_i$, where $i \in \mathbb{N}$, such that, when the contract is set to $p_{\omega}=1$, then the best response of a type $\theta_i$ is $a_i$.
We construct such sequence as follows: for each $i$, we set
\begin{itemize}
\item $F_{a_i,\omega}= \frac{1}{2i} $;
\item $c_{a_i}=\frac{1}{4i^2}$;
\item  types $\theta_i=i$.
\end{itemize}
Then, an agent of type $i$ will play the action maximizing their utility:
\[
i=\arg \max_{j \in \mathbb{N}} \left[ p_\omega F_{a_j,\omega} -\theta_{i}c_{a_j} \right] =\arg \max_{j \in \mathbb{N}} \left[ \frac{1}{2j} -\frac{i}{4j^2} \right], 
\]
which is the action corresponding to its type, as desired.
\begin{figure}[!t]
    \centering
    \scalebox{0.7}{\input{imgs/setcover}}
    \caption{Stylized instance corresponding to the \textsc{Set-Cover} instance: $S_1=\{1,2\},S_3=\{2\},S_3=\{1,3\},S_4=\{3\}$ with universe $E=\{1,2,3\}$. An example of a cover is $\Scal^\star=\{S_2,S_3\}$. The table corresponds to payments $p_{\omega_{S_2}}=p_{\omega_{S_3}}=1$ while all the others are $0$. In each cell, there is the utility of an agent of type $i$ (indexed by rows) of playing an action of type $a_{j,S}$ (indexed by columns). We highlighted in \textcolor{niceRed}{red} the actions corresponding to the outcomes $\omega_{S_2}$ and $\omega_{S_3}$, and in \textcolor{typ_blue}{blue} the action played by each agent.}
    \label{fig:instance}
\end{figure}
This sequence is used when reducing from \textsc{Set-Cover} by replicating this sequence for each $S\in\Scal$ if $i\in S$. Formally, we have actions $a_{i,S}$ for each $i\in S$ and $S\in\Scal$, and outcome $\omega_S$ for each $S\in\Scal$. Then, $F_{a_{i,S},\omega_{S}}=\frac{1}{2i}$ and $c_{a_{i,S},\omega_S}=\frac{1}{4i^2}$ as above. The properties of the sequence guarantee that an agent $i$ plays an action $a_{i,S}$ for some $S\in\Scal$ if the principal picked payment $p_{\omega_S}=1$. Thus for a cover $\Scal^\star$ we can choose payments $p_{\omega_S}=\mathbb{I}(S\in\Scal^\star)$, and each agent of type $i$ has an action $a_{i,S}$ to play.
\Cref{fig:instance} provides depiction of the sequence associated with a simple \textsc{Set-Cover} instance. 
In the next section, we provide the formal proof of \Cref{thm:reduction}.

\subsection{Proof of \Cref{thm:reduction}}

	We reduce from \textsc{Set-Cover}.
	Given a set of elements $E$ and a set of sets $\Scal \subseteq 2^E$, we reduce from the decision problem of determining if there exists a set cover of size $k$, i.e., a set $\Scal^\star\subseteq \Scal$ of size $k$ such that 
	$\bigcup_{S \in \Scal^\star} S=E$. This problem is well known to be NP-Hard \cite{Karp1972}.
	
	We show that if there exists a set cover of size $k$, then principal's utility from the optimal contract is at least $\ell$ where $\ell$ will be defined in the following, while if all the set covers have size at least $k+1$ then the utility is at most $\ell-|E|^{-15}|S|^{-1}$.

	\paragraph{\textbf{Construction}}

	We identify $E$ with the set $\{1,2,\ldots, n\}$, where $|E|=n$. 
    Moreover, we let $m=|\mathcal{S}|$.
	For the ease of presentation, we define  the values $\rho=n^{-6}$, $\eta=n^{-2}$, $\epsilon=n^{-8}m^{-1}$, and $\mu=n^{-9}m^{-1}$, which are all polynomial in the instance. 
	Then, given an instance $(E,\Scal)$ of set cover, the instance of single-dimensional Bayesian contract design is defined as follows:
	\begin{itemize}
		\item Outcomes
		\begin{itemize}
			\item an outcome $\omega_{S}$ for each $S \in \Scal$,
			\item two additional outcomes $\omega^\star$ and $\bar \omega$.
		\end{itemize}
		
		\item Rewards
		\begin{itemize}
			\item $r_{\omega_S}=0$ for each $S \in \Scal$,
			\item $r_{\omega^\star}=\frac{1}{n}$ and $r_{\bar \omega}=0$.
		\end{itemize}
		
		\item Actions
		\begin{itemize}
			\item an action $a_{i,S}$ for each set $S\in \Scal$ and for each $i\in S$,
			\item an action $\bar a_{i,S}$ for each set $S\in\Scal$ and for each $i\in S$, 
			\item two additional actions $a^\star$ and $a_0$.
		\end{itemize}
		
		\item Actions cost
		\begin{itemize}
			\item cost $c_{a_{i, S}}=\frac{1}{4 i^2}\mu $ for each $S\in\Scal$ and $i\in S$,
			\item cost $c_{\bar a_{i, S}}=\frac{1}{4 i^2}\mu  (1-\eta)$ for each $S\in\Scal$ and $i\in S$,
			\item $c_{a^\star}=1$ and $c_{a_0}=0$.
		\end{itemize}
		
		\item Outcomes distribution (all unspecified probabilities are set to $0$)
		\begin{itemize}
			\item  $F_{a_{i,S},\omega_S}=\frac{1}{2i}\mu$,  $F_{a_{i,S}, \omega^\star}=\frac{1}{i}\mu$, and $F_{a_{i,S},\bar \omega}=1-\frac{3}{2i}\mu$ for each set $S\in \Scal$ and $i\in S$,
		\item $F_{\bar a_{i,S},\omega_S}=\frac{1}{2i}(1-\frac\eta2)\mu$, $F_{\bar a_{i,S},\bar \omega}=1-\frac{1}{2i}(1-\frac\eta2)\mu$ for each set $S\in\Scal$ and $i\in S$,
			\item $F_{a^\star,\omega_S}=\epsilon$ for each $S\in\Scal$ and $F_{a^\star,\omega^\star}=1-m\epsilon$,
			\item $F_{a_0, \bar\omega}=1$.
		\end{itemize}
		
		\item Types 
		\begin{itemize}
			\item a type $\theta_i=\frac{i}{n}$ for each $i\in E$,
            \item  an additional type $\theta=0$. 
		\end{itemize}
		
		\item Distribution over types
		\begin{itemize}
			\item $\gamma_{\theta_i}=\frac{1-\rho}{n}$ for each $i \in E$ where $\theta_i=\frac{i}{n}$,
            \item  $\gamma_{\theta_0}=\rho$ where $\theta_0=0$.
            \end{itemize}

		%
		%
		%
	\end{itemize}
    
	\paragraph{\textbf{If analysis.}}
	Suppose that there exists a set cover $\Scal^\star$ of size $k$.
	We build a contract with utility at least 
	\[
            \ell:= \frac{1-\rho}{2n^2} \mu \sum_{i\in E} \frac{1}{i} + \frac{\rho}{n} (1-m\epsilon -\epsilon k).  
        \]
	The contract is defined as follows.	For each set $S\in \Scal^\star$, set
	$p_{\omega_S}=\frac{1}{n}$.
	Moreover, let $p_{\omega^\star}=p_{\bar \omega}=0$ and $p_{\omega_S}=0$ for all $S\notin \Scal^\star$.
	
    We start computing the best responses for each possible agent's type. Then, we bound the principal's expected utility according to the best responses.
    In \Cref{app:reduction} we show that
    \begin{itemize}
        \item For any type $\theta_i$ the best response is any action $\tilde a_i\in\{a_{i,S}\}_{S\in\Scal^\star}$ (\Cref{lem:BR_IF}).\footnote{Technically, we only show that $\tilde a_i$ is IC, \emph{i.e.}, $\tilde a_i\in \Bcal^{\theta_i}(p)$. However, since ties are broken in favor of the principal, if the agent plays a different action, the principal obtains an even larger utility. }
        \item For the type $\theta_0$ the best response is $a^\star$ (\Cref{lem:BR_IF2}).
    \end{itemize}

	%
	
	Knowing the best-responses for each type, we can compute the principal's utility committing to the previously defined $p$:
	\begin{align*}
		&\sum_{i \in E} \gamma_{\theta_i} \sum_{\omega\in \Omega}F_{\tilde a_{i},\omega}(r_\omega-p_\omega) + \gamma_{\theta_0} \sum_{\omega\in\Omega }F_{a^\star,\omega}(r_\omega - p_\omega)\\
		& \hspace{0.5cm}=\sum_{i\in E} \frac{1-\rho}{n} \left[F_{\tilde a_{i},\omega_S}(r_{\omega_S}-p_{\omega_s})+F_{\tilde a_{i},\omega^\star}(r_{\omega^\star}-p_{\omega^\star})\right]\tag*{\makebox[1pt][r]{(Utility for types $\theta_i$)}}
		\\
		& \hspace{2cm}+ \rho\left(F_{a^\star,\omega^\star}(r_{\omega^\star}-p_{\omega^\star})+\sum_{S\in\Scal}F_{a^\star,\omega_S}(r_{\omega_S}-p_{\omega_S})\right)\tag*{\makebox[1pt][r]{(Utility for type $\theta_0$)}}\\
		&\hspace{0.5cm}=\frac{1-\rho}{n}\sum_{i\in E}\left[\frac{1}{2i}\mu\left(0-\frac{1}{n}\right)+\frac{1}{i}\mu\left(\frac{1}{n}-0\right)\right]+\rho\left((1-m\epsilon)\left(\frac{1}{n}-0\right)+|\Scal^\star|\epsilon\left(0-\frac{1}{n}\right)\right)\\
		&\hspace{0.5cm}=\frac{1-\rho}{n^2} \mu \sum_{i\in E} \frac{1}{2i}+\frac{\rho}{n}(1-m\epsilon-\epsilon|\Scal^\star|)\\
		&\hspace{0.5cm}=\frac{1-\rho}{2n^2}  \mu\sum_{i\in E}\frac{1}{i}+\frac{\rho}{n}(1-m\epsilon-\epsilon k)\\
        &\hspace{0.5cm}=:\ell.
	\end{align*}
    This completes the first part of the proof, lowerbounding the principal's utility when there exists an set cover of size $k$.

	\paragraph{\textbf{Only if analysis.}}
	Suppose that all set covers have size at least $k+1$.
	We show that any contract provides to the principal a utility of at most $\ell-n^{15}m^{-1}$.
	Let $p\in\Reals^{|\Omega|}_{\ge 0}$ be any contract.
	
    We partition the set of elements $E$ into three sets.
    \begin{itemize}
	\item $E_1\subseteq E$ includes the elements $i\in E$ such that the best response of type $\theta_i$ is $a_{i,S}$ for some $S \in \Scal$,
	\item $E_2\subseteq E$ includes  the elements $i\in E$ such that the best response of type $\theta_i$ is $a_{j,S}$ with $j\neq i$ and $S \in \Scal$,
	\item $E_3$ includes the remaining elements, \emph{i.e.}, $E_3=E \setminus (E_1\cup E_2)$.
    \end{itemize}
    
    In \Cref{app:reduction} we prove the following bounds on the principal utility conditioning on each of these cases.
    \begin{itemize}
        \item For any $i\in E_1$ the principal utility when the realized type is $\theta_i$ is at most $\frac{1}{2in}\mu+\frac{1}{i}p_{\omega^\star}\mu\left(\frac2\eta-1\right)$ (\Cref{lem:onlyif1}),
        \item For any $i\in E_2$ the principal utility when the realized type is $\theta_i$ is at most $\mu\left(\frac{1}{2in}-\frac{1}{8n^4}+\frac{2}{\eta}p_{\omega^\star}\right)$ (\Cref{lem:onlyif2}),
        \item For any $i\in E_3$ the principal utility when the realized type is $\theta_i$ is negative (\Cref{lem:onlyif3}),
        \item When the realized type is $\theta_0$ the principal utility is $-\epsilon\sum_{S\in\Scal}p_{\omega_S}+(1-m\epsilon)(\frac{1}{n}-p_{\omega^\star})$ (\Cref{lem:onlyif4})
    \end{itemize}

    We can then combine these results to show an upper bound on the principal utility of
    \begin{equation}\label{eq:upperboundutils}
    \frac{1-\rho}{n}\mu\left[\sum_{i\in E_1}\frac{1}{2in}+\sum_{i\in E_2}\left(\frac{1}{2in}-\frac{1}{8n^4}\right)\right]+\frac{1}{n}\rho(1-m\epsilon)-\frac{1}{n}\epsilon\rho|\bar\Scal|,
    \end{equation}
    where $\bar \Scal:=\{S\in\Scal: p_{\omega_S}\ge\frac{1}{n}-\frac{4}{\eta}p_{\omega^\star}\}$.
    This result is proved in \Cref{lem:onlyif5}.

    Now, we compute an upperbound on the principal's utility. Intuitively, we show that the principal utility increases if we ``move'' all the elements of $E_2\cup E_3$ to $E_1$ and increase the payment by a $\frac{1}{n}\epsilon\rho |E_2\cup E_3|$ additive factor.
    Formally, we observe that we can ``move'' a type from $E_2$ to $E_1$ by adding a $\frac{1}{n}\rho\epsilon$ payment since:
	\begin{align*}
		\frac{1-\rho}{n} \mu \frac{1}{8n^4}&\ge \frac{\mu}{16n^5}=\frac{1}{16n^{14}m}\\
		&\ge\frac{1}{n^{15}m}\tag{for $n$ large enough}\\
		&=\frac{1}{n}\rho\epsilon,
    \end{align*}
	and thus, starting from \Cref{eq:upperboundutils}, we find that the principal utility is upper bounded by
	\[
	\frac{1-\rho}{n}\mu\sum_{i\in E}\frac{1}{2in}+\frac{1}{n}\rho(1-m\epsilon)-\frac{1}{n}\epsilon\rho(|\bar\Scal|+|E\setminus E_1|).
	\]

    We complete the proof showing that $|\bar\Scal|+|E\setminus E_1|$ is a upper bound on the size of the smallest set cover, and hence by assumption  $|\bar\Scal|+|E\setminus E_1|\ge k+1$.
    To do so, we complete the set $\bar \Scal$ to make it a set cover of $E$ in the following way. For each element $v$ in $E\setminus E_1$ we pick a set $s(v)\in\Scal$ such that $v\in s(v)$ and we claim that $\Scal^\star:=\bar\Scal\bigcup\left( \bigcup_{v\in E\setminus E_1}s(v)\right)$ is a cover of $E$.
    Indeed, by construction the set $\bar S$ is a cover of $E_1$ and $\bigcup_{v\in E\setminus E_1}s(v)$ is a cover of $E\setminus E_1$. Thus, $\Scal^\star$ is a cover of $E$ and it has size at least $k+1$ by assumption. 
    
    Now, note that $|\bigcup_{v\in E\setminus E_1}s(v)|\le |E\setminus E_1|$. The proof is concluded by observing that $|\bar \Scal|+|E\setminus E_1|\ge |\Scal^\star|\ge k+1$ and thus the principal utility is upper bounded by
	\[
	\frac{1-\rho}{2n^2}\mu\sum_{i\in E}\frac{1}{i}+\frac{\rho}{n}\left(1-m\epsilon-\epsilon\rho(k+1)\right)=\ell-\frac{1}{n}\epsilon\rho=\ell-\frac{1}{n^{15}m},
	\]
    concluding the proof.

%% file: imgs/setcover.tex
\tikzset{every picture/.style={line width=0.75pt}} 

\begin{tikzpicture}[x=0.75pt,y=0.75pt,yscale=-1.2,xscale=1.2]

\input{settings/colors}
\colorlet{highcell}{typ_blue}
\colorlet{cell}{mgray!80!}

\draw  [fill=cell  ,fill opacity=0.3 ] (130,80) -- (190,80) -- (190,140) -- (130,140) -- cycle ;
\draw  [fill=cell  ,fill opacity=0.3 ] (190,80) -- (250,80) -- (250,140) -- (190,140) -- cycle ;
\draw  [fill=cell  ,fill opacity=0.3 ] (250,80) -- (310,80) -- (310,140) -- (250,140) -- cycle ;
\draw  [fill=cell  ,fill opacity=0.3 ] (310,80) -- (370,80) -- (370,140) -- (310,140) -- cycle ;
\draw  [fill=cell  ,fill opacity=0.3 ] (370,80) -- (430,80) -- (430,140) -- (370,140) -- cycle ;
\draw  [fill=cell  ,fill opacity=0.3 ] (430,80) -- (490,80) -- (490,140) -- (430,140) -- cycle ;
\draw  [fill=cell  ,fill opacity=0.3 ] (130,140) -- (190,140) -- (190,200) -- (130,200) -- cycle ;
\draw  [fill=cell  ,fill opacity=0.3 ] (190,140) -- (250,140) -- (250,200) -- (190,200) -- cycle ;
\draw  [fill=cell  ,fill opacity=0.3 ] (250,140) -- (310,140) -- (310,200) -- (250,200) -- cycle ;
\draw  [fill=cell  ,fill opacity=0.3 ] (310,140) -- (370,140) -- (370,200) -- (310,200) -- cycle ;
\draw  [fill=cell  ,fill opacity=0.3 ] (370,140) -- (430,140) -- (430,200) -- (370,200) -- cycle ;
\draw  [fill=cell  ,fill opacity=0.3 ] (430,140) -- (490,140) -- (490,200) -- (430,200) -- cycle ;
\draw  [fill=cell  ,fill opacity=0.3 ] (130,200) -- (190,200) -- (190,260) -- (130,260) -- cycle ;
\draw  [fill=cell  ,fill opacity=0.3 ] (190,200) -- (250,200) -- (250,260) -- (190,260) -- cycle ;
\draw  [fill=cell  ,fill opacity=0.3 ] (250,200) -- (310,200) -- (310,260) -- (250,260) -- cycle ;
\draw  [fill=cell  ,fill opacity=0.3 ] (310,200) -- (370,200) -- (370,260) -- (310,260) -- cycle ;
\draw  [fill=cell  ,fill opacity=0.3 ] (370,200) -- (430,200) -- (430,260) -- (370,260) -- cycle ;
\draw  [fill=cell  ,fill opacity=0.3 ] (430,200) -- (490,200) -- (490,260) -- (430,260) -- cycle ;
\draw  [draw opacity=0] (130,40) -- (490,40) -- (490,80) -- (130,80) -- cycle ;
\draw  [draw opacity=0] (490,80) -- (530,80) -- (530,260) -- (490,260) -- cycle ;
\draw  [color={rgb, 255:red, 245; green, 166; blue, 35 }  ,draw opacity=1 ][fill=highcell  ,fill opacity=0.3 ][line width=2.25]  (190,80) -- (250,80) -- (250,140) -- (190,140) -- cycle ;
\draw  [color={rgb, 255:red, 245; green, 166; blue, 35 }  ,draw opacity=1 ][fill=highcell  ,fill opacity=0.3 ][line width=2.25]  (310,140) -- (370,140) -- (370,200) -- (310,200) -- cycle ;
\draw  [color={rgb, 255:red, 245; green, 166; blue, 35 }  ,draw opacity=1 ][fill=highcell  ,fill opacity=0.3 ][line width=2.25]  (370,200) -- (430,200) -- (430,260) -- (370,260) -- cycle ;
\draw  [draw opacity=0] (130,260) -- (190,260) -- (190,300) -- (130,300) -- cycle ;
\draw  [draw opacity=0] (190,260) -- (250,260) -- (250,300) -- (190,300) -- cycle ;
\draw  [draw opacity=0] (250,260) -- (310,260) -- (310,300) -- (250,300) -- cycle ;
\draw  [draw opacity=0] (310,260) -- (370,260) -- (370,300) -- (310,300) -- cycle ;
\draw  [draw opacity=0] (370,260) -- (430,260) -- (430,300) -- (370,300) -- cycle ;
\draw  [draw opacity=0] (430,260) -- (490,260) -- (490,300) -- (430,300) -- cycle ;
\draw  [draw opacity=0] (90,80) -- (130,80) -- (130,140) -- (90,140) -- cycle ;
\draw  [draw opacity=0] (90,140) -- (130,140) -- (130,200) -- (90,200) -- cycle ;
\draw  [draw opacity=0] (90,200) -- (130,200) -- (130,260) -- (90,260) -- cycle ;

\draw (160,280) node  [font=\Large]  {$a_{1,S_{1}}$};
\draw (220,280) node  [font=\Large,color={rgb, 255:red, 208; green, 2; blue, 27 }  ,opacity=1 ]  {$a_{1,S_{3}}$};
\draw (280,280) node  [font=\Large]  {$a_{2,S_{1}}$};
\draw (340,280) node  [font=\Large,color={rgb, 255:red, 208; green, 2; blue, 27 }  ,opacity=1 ]  {$a_{2,S_{2}}$};
\draw (400,280) node  [font=\Large,color={rgb, 255:red, 208; green, 2; blue, 27 }  ,opacity=1 ]  {$a_{3,S_{3}}$};
\draw (460,280) node  [font=\Large]  {$a_{3,S_{4}}$};
\draw (310,60) node   [font=\Large] {$a_{j,S}$};
\draw (513,170) node  [font=\Large]  {$i$};
\draw (220,170) node  [font=\Large]  {$\frac{1}{2j} -\frac{2}{4j^{2}}$};
\draw (160,170) node  [font=\Large]  {$-\frac{2}{4j^{2}}$};
\draw (400,170) node  [font=\Large]  {$\frac{1}{2j} -\frac{2}{4j^{2}}$};
\draw (460,170) node  [font=\Large]  {$-\frac{2}{4j^{2}}$};
\draw (280,170) node  [font=\Large]  {$-\frac{2}{4j^{2}}$};
\draw (340,170) node  [font=\Large]  {$\frac{1}{2j} -\frac{2}{4j^{2}}$};
\draw (110,110) node  [font=\Large] {$1$};
\draw (110,170) node  [font=\Large] {$2$};
\draw (110,230) node  [font=\Large] {$3$};
\draw (460,110) node  [font=\Large]  {$-\frac{1}{4j^{2}}$};
\draw (460,230) node  [font=\Large]  {$-\frac{3}{4j^{2}}$};
\draw (400,230) node  [font=\Large]  {$\frac{1}{2j} -\frac{3}{4j^{2}}$};
\draw (400,110) node  [font=\Large]  {$\frac{1}{2j} -\frac{1}{4j^{2}}$};
\draw (280,230) node  [font=\Large]  {$-\frac{3}{4j^{2}}$};
\draw (160,230) node  [font=\Large]  {$-\frac{3}{4j^{2}}$};
\draw (220,230) node  [font=\Large]  {$\frac{1}{2j} -\frac{3}{4j^{2}}$};
\draw (340,110) node  [font=\Large]  {$\frac{1}{2j} -\frac{1}{4j^{2}}$};
\draw (220,110) node  [font=\Large]  {$\frac{1}{2j} -\frac{1}{4j^{2}}$};
\draw (280,110) node  [font=\Large]  {$-\frac{1}{4j^{2}}$};
\draw (160,110) node  [font=\Large]  {$-\frac{1}{4j^{2}}$};
\draw (340,230) node  [font=\Large]  {$\frac{1}{2j} -\frac{3}{4j^{2}}$};

\end{tikzpicture}

%% file: src/learning.tex
\section{Learning Contracts with Single-Dimensional Costs and Continuous Types}\label{sec:learning}


We proved in \Cref{sec:hardness} that there are no polynomial-time algorithms that compute $\epsilon$-optimal contracts under additive approximations. Well-know reductions from offline to online problems (see, \emph{e.g.}, \citet{roughgarden2019minimizing}) essentially preclude the existence of efficient algorithms that achieve regret polynomial in the instance size and sub-linear in $T$ in this setting, unless $\NP\subseteq\RP$. 
Therefore, we shift our focus to the problem of achieving sub-linear regret with \emph{inefficient} algorithms, prioritizing the learning aspects over computational efficiency.

The main result of this section is a learning algorithm that provides regret of $\widetilde O(\poly(n,m)T^{3/4})$ assuming that the types are stochastic and single-dimensional.
Then, we extend our result to provide a sample complexity bound of $\widetilde O\left(\eta^{-4}\poly(m,n)\right)$ to learn $\eta$-optimal contracts.

\subsection{Setting and Assumptions}

In our learning setup, we introduce two standard assumptions.
First, as is common when handling continuous types \citep{alon2023bayesian, alon2021contracts}, we make some minimal assumptions on the underlying distributions. In particular, we consider distributions which admits density bounded by some constant $\beta$.

\begin{assumption}\label{ass:boundedDensity}
    The type distribution $\Gamma$ admits density $f_\Gamma$ such that  $\sup_{\theta\in[0,1]}f_\Gamma(\theta)\le \beta$.
\end{assumption}

Our second assumption is on the set of possible contracts available to the principal. Following previous work on learning in contract design problems \citep{bacchiocchi2023learning,zhu2022online,chen2024bounded,ho2014adaptive}, we restrict our focus to bounded contracts.
\begin{assumption}\label{ass:boundedContract}
	The principal can post only contracts in $p\in [0,1]^m$.
\end{assumption}
Our results will easily extend to bounded contracts $p\in [0,C]^m$, by adding a linear dependence on $C$ in the regret bounds. Intuitively, the maximum possible payment determines the range of possible principal's utility. This affects the variance of the observed random variables and the regret bounds.
With slight abuse of notation, we now use $\OPT$ to denote the value of the optimal bounded contract in $[0,1]^m$.

\subsection{A No-Regret Algorithm}

In this section, we design a regret-minimization algorithm for single-dimensional contract design problem defined, as defined in \Cref{sec:prelimLearning}. Then, in the next section, we extend our results to derive sample complexity bounds.

There are two main challenges in obtaining no-regret algorithms for our problem.

\paragraph{Continuous decision space} The primary and most evident  challenge is that the decision space (the $m$-dimensional hypercube) is continuous, and that the function $p\mapsto \Up(p)$ lacks useful structural properties (indeed, it is neither concave nor Lipschitz). As a result, most existing techniques for multi-dimensional contract design, which operate directly on this decision space and do not take advantage of single-dimensional types, become inapplicable~\cite{zhu2022online,bacchiocchi2023learning,ho2014adaptive}.

\paragraph{Infinite types} The second challenge arises from having infinitely many possible types, as the support is $\Theta=[0,1]$. This renders standard approaches for Stackelberg-like games inapplicable. Indeed, these approaches for learning in games with commitments crucially rely on having a finite number of types.
For example, \citet{balcan2015commitment} employ an explore-then-commit algorithm that leverages barycentric spanners to estimate the values of exponentially many possible decision while playing only a polynomial number of them (equal to the number of types). As is typical in explore-then-commit algorithms, this approach achieves $    \widetilde O(T^{2/3})$ regret and a polynomial dependence on the number of types. Hence, similar approaches are doomed to fail in settings with continuous types.
As a further example, \citet{bernasconi2023optimal} avoid the exploration phase by reducing the problem to linear bandit optimization on a suitable small-dimensional space with dimension $d$ equal to the number of types $|\Theta|$. This approach yields regret bounds that depend only on the dimensionality of the space rather than the number of actions, resulting in $O(\poly(d)\sqrt{T})$ regret.
Both approaches perform poorly when the number of types is infinite. This is not just a technical limitation: there is strong evidence that as the number of types grows large, one cannot hope to obtain no-regret algorithms in general. For instance, \citet[Theorem~7.1]{balcan2015commitment} show that when the number of agent types is not bounded, then the regret is necessarily linear. For contract design, \citet{zhu2022online} provide similar results.

We address these challenges by introducing an approximate instance with finitely many types, using the same discretized grid as in \Cref{sec:PTAS}.
Here, we have the additional challenge of showing that such grid provides a good approximation despite not knowing the type distribution $\Gamma$. 
As we will discuss in details in the following, this is \emph{not} true for general distributions.
However, we show that if the underlying types distribution has bounded density, then for \emph{any} contract $p$ the expected principal's utilities with continuous and with discretized types are close.
Formally, we prove the following:

\begin{restatable}{lemma}{lemmaClose}\label{lem:valClose}
    For any type distribution $\Gamma$ and any $\epsilon>0$, let $\Theta_\epsilon\coloneqq\{\epsilon\cdot(i-\frac{1}{2})\}_{i\in\lceil1/\epsilon\rceil}$ and define $\gamma$ supported on $\Theta_\epsilon$ as 
\[\gamma_{\theta}\coloneqq \int_{\theta-\epsilon/2}^{\theta+\epsilon/2} f_\Gamma(\tilde\theta) d\tilde\theta\quad\textnormal{ for any }\theta\in\Theta_\epsilon.\]
Then, under \Cref{ass:boundedDensity} and \Cref{ass:boundedContract}, for any $p \in [0,1]^m$, it holds that
    \[  
        \left|\mathbb{E}_{\theta\sim\Gamma}[\Up(p,\theta)]-\sum_{\theta\in \Theta_\epsilon}  \gamma_\theta \ \Up(p, \theta) \right|\le 2\beta n \epsilon.
    \] 
\end{restatable}

\begin{remark}
    \Cref{lem:valClose} is fundamentally different from \Cref{lem:linearization} and other results that connect $\epsilon$-IC and IC contracts  to ``robustify'' contracts~\citep{dutting2021complexity,zhu2022online,bacchiocchi2023learning, bernasconi2024regret}. 
    These existing results apply to arbitrary distributions but only provide a lower bound on the achieved utility, often significantly underestimating the principal's utility, and rendering them inapplicable in our approach. 
    In contrast, our result crucially relies on \Cref{ass:boundedDensity} and establishes both upper and lower bounds on the utility of any given contract.
\end{remark}

Hence, the above lemma addresses the 
challenge of infinitely many types. In particular, by incurring a discretization error of $O(\beta n\epsilon T)$, we can restrict to a finite set of discretized types of size $O(1/\epsilon)$.
%
Next, we deal with the problem of having an infinite decision space. Our first result partially resolves this by reducing the set of possible decisions from infinite to exponentially many. 
To achieve this, we leverage a powerful characterization of optimal contracts in Bayesian contract design that is independent of the underlying type distribution. This is a standard approach in Stackelberg-like games and is, for example, one of the main tools employed by \citet{balcan2015commitment} for solving learning problems in Bayesian Stackelberg games. 
A similar result also exists in contract design, albeit for \emph{discrete types}. Indeed, \citet{castiglioni2022bayesian,guruganesh2021contracts} showed that, regardless of the underlying distribution over discrete types, the optimal contract always belongs to a finite set of possible contracts $\Pcal$.\footnote{While the original result is for unbounded contracts, it trivially extend to $[0,1]^m$ by adding the hyperplanes $p_\omega\le1$ for each $\omega\in \Omega$.}

\begin{theorem}[Essentially \cite{guruganesh2021contracts, castiglioni2022bayesian}]\label{thm:toFinite}
    For every finite set of types $\Theta$, there exists a finite set of contracts $\Pcal\subset [0,1]^m$ with $|\Pcal|= \poly\left((n,m,|\Theta|)^m\right)$ such that
    \[
    \max_{p\in \Pcal} \sum_{\theta\in \Theta}  \gamma_\theta \Up(p, \theta) = \max_{p\in [0,1]^m} \sum_{\theta\in \Theta}  \gamma_\theta \Up(p, \theta)\quad  \forall \gamma \in \Delta_{\Theta}.
    \]
\end{theorem}

Now, we show how to combine the results in \Cref{lem:valClose} and \cref{thm:toFinite}.
In \Cref{lem:valClose}, we showed that reducing to discretized types provides a good approximation of the reward with the true continuous types, and \Cref{thm:toFinite} shows how to restrict to a finite number of possible contracts $\Pcal$. A simple approach would be to reduce our problem to a finite MAB in which $\Pcal$ is the finite set of arms. This would provide a regret upper bound of $\sqrt{|\Pcal| T}+O(\epsilon n\beta T)$. Since $|\Pcal|=\poly((n,m,1/\epsilon)^m)$, this yields a regret upper bound that, although sub-linear in 
$T$ (for some choice of $\epsilon$), is  exponential in the instance size. 

\citet{bernasconi2023optimal} provided a general framework to solve similar problems and remove the dependency on the size of the decision space.
Such a framework reduces the problem of online learning in Stackelberg-like games with finite types to linear bandits. Their main idea is to map the decisions into a ``utility space'' instead of the original decision space, and then map the decisions from the utility space back to the original one. The advantage is that the problem in the utility space is linear, and thus it scales only with the dimension rather than with the number of actions.
Employing such reduction, \citet{bernasconi2023optimal} provided regret bounds of order $O(\sqrt{T})$ that scale polynomially with the number of types rather than with the number of possible decisions.
Notice that their reduction cannot directly be employed in in our setting since the discrete type instance is only an approximation of the true one, breaking the linear structure of the utility space.

Since we cannot recover linearity in the utility space, our key idea is to continue working in such space while allowing the model to be only approximately linear. Crucially, because the linearity assumption is violated, we must use algorithms designed for a class of almost-linear bandits known as Misspecified Linear Bandits (MLB)~\cite{ghosh2017misspecified}.

\begin{definition}[Misspecified linear bandits~\citep{lattimore2020learning}]\label{def:MLB}
    At each round $t \in [T]$, the learner chooses an action $x_t$ from a finite set of action $X\subseteq [0,1]^d$ with $|X|=k$. The expected reward of an action $x$ is $\mu(x)=\langle x,\phi\rangle+ y(x) \in [-1,1]$, where $ |y(x)|\le\alpha$ for a misspecification parameter $\alpha>0$. The observed rewards $\mu_t$ are i.i.d.~and $O(1)$-subgaussian with mean $\mathbb{E}[\mu_t]=\mu(x_t)$.
    Finally, the regret is defined as
    \[
    R_T= T \cdot\max_{x\in X} \mu(x)- \mathbb{E}\left[\sum_{t \in [T] }  \mu(x_t)\right],
    \]
    where the expectation is on the possible randomization of the algorithm.
\end{definition}

We will use an algorithm proposed by \citet{lattimore2020learning} as a black-box oracle for this setting.

\begin{theorem}[\cite{lattimore2020learning}]\label{thm:misspecified}
    There exists an algorithm called $\PE$ for misspecified linear bandits which guarantees
    \[
        R_T=O\left(\sqrt{dT\log(Tk)}+ \alpha T\sqrt{d} \log(T)\right). 
    \]
\end{theorem}
Crucially, $\PE$ does not need to know the misspecification parameter $\alpha$ in advance.

Now, we are ready to combine all these components to provide our reduction, whose pseudo-code can be found in Algorithm~\ref{alg:reduction}.
\setlength{\algomargin}{1.5em}
\begin{algorithm}[t]
    \nl Inputs: $\epsilon>0$, $\Pcal$\;
    \nl Set $\Rcal$ as $\PE$ and $X=\nu_\epsilon(\Pcal)$\;
    \For{$t = 1, \ldots, T$}{	
    \nl Set $x_t$ as prescribed by $\Rcal$\;\nllabel{line:decision}
    \nl Commit to any $p_t$ such that $x_t = \nu_\epsilon(p_t)$\;\nllabel{line:invert}
    \nl Agent type $\theta_t$ is drawn according to $\Gamma$\;
    \nl Action $a_t=b^{\theta_t}(p_t)$ is played\;\nllabel{line:br}
    \nl Outcome $\omega_t\sim F_{a_t}$ is observed\;\nllabel{line:outcome}
    \nl Reward $r_{\omega_t}-p_{\omega_t}$ is received and it is used to update $\Rcal$\;
    }
    \caption{Regret minimization for single-dimensional contract design}\label{alg:reduction}
\end{algorithm}

%
%

First, we define the set of actions to be given in input to the regret minimizer. Set $\epsilon=1/\sqrt{T}$.
For any $p \in \Pcal$ (see \Cref{thm:toFinite} for a definition of $\Pcal$), let 
\[
\nu_\epsilon(p)=[\Up(p,\theta)]_{\theta \in \Theta_\epsilon}\quad\text{and}\quad\nu_\epsilon(\Pcal)= \bigcup_{p \in \Pcal} \nu_\epsilon(p).
\]

Note that $\nu_\epsilon(\Pcal)\subset \Reals^{|\Theta_\epsilon|}$ is a finite set of cardinality $|\nu_\epsilon(\Pcal)|=|\Pcal|=\poly((n,m,1/\epsilon)^m)$.
The next decision $x_t$ is chosen in the utility space by the MLB regret minimizer $\Rcal$ (\Cref{line:decision}), while actions $p_t$ are made in the actual decision space of contracts inverting the function $\nu_\epsilon(\cdot)$ (\Cref{line:invert}). After the principal commits to contract $p_t$ the agent selects the best response and the outcome $\omega_t$ is selected with probability $F_{a_t,\omega_t}$ (\Cref{line:br} and \Cref{line:outcome}). Finally, the reward received $r_{\omega_t}-p_{\omega_t}$ is used to update the regret minimizer $\PE$.

The following theorem proves that \Cref{alg:reduction} provides a $\widetilde O(\beta \cdot\poly(n,m)T^{3/4})$ regret.
In the proof, we show that the discretized problem satisfies the conditions of \Cref{def:MLB}. Then, we prove that the regret guarantees of the discretized problem translate to the regret guarantees of the original problem. This holds since the optima of the two problems are close. 

\begin{theorem}\label{thm:regret}
	Under \Cref{ass:boundedDensity} and \Cref{ass:boundedContract}, \Cref{alg:reduction} guarantees
	 \[R_T=\widetilde O(\beta \cdot\poly(n,m)T^{3/4}).\]
\end{theorem}

\begin{proof}
	Let $X=\nu_\epsilon(\Pcal)$. First, we show that the structure of the reward observed by the regret minimizer for MLS satisfies \Cref{def:MLB}.
	Clearly, the reward observed is always in $[-1,1]$ since both $r_\omega$ and $p_\omega$ are bounded in $[0,1]$. Next, we show that there exists a good linear approximation of $\mu(x)$ for any $x\in X$. In particular, we will show that $\mu(x)$ is close to $\langle x, \gamma\rangle$ where $\gamma_{\theta}=\int_{\theta-\epsilon/2}^{\theta+\epsilon/2} f_\Gamma(\theta)d\theta$, which follows from \Cref{lem:valClose}:
    
    \begin{align*}
        |\mu(x)-\langle x,\gamma\rangle|&=\left|\int_{0}^1f_\Gamma(\theta)\Up(p,\theta)d\theta-\sum_{\theta\in\Theta_\epsilon}\gamma_\theta\Up(p,\theta)\right|\le 2\beta n\epsilon.
    \end{align*}
    Therefore, the misspecification parameter is $\alpha=2\beta n\epsilon$.

	Then, by \Cref{thm:misspecified}, we conclude that
	\begin{align*}
		 T\cdot \max_{x\in X} \mu(x)- \mathbb{E}\left[ \sum_{t\in [T]}  \mu(x_t)\right]& \le O(\sqrt{|\Theta_\epsilon|T}\,\log(T (\poly(n,m,|\Theta_\epsilon|)^m))+ 2n\beta\epsilon T \sqrt{|\Theta_\epsilon|}\, \log(T))\\
		 & \le  \widetilde O\left(\beta\cdot \poly(m,n) T^{3/4}\right),
	\end{align*}
	where the last inequality follows by setting $\epsilon=1/\sqrt{T}$,  and the expectation is over the randomness of the algorithm.
	
	
	Hence, we are left with bounding the difference in utility between $\OPT$ and the optimum in the discretized space, which is $\max_{x \in X} \mu(x)$.
        Since the contract is bounded, we know that there exists an optimal bounded contract $p^*\in \arg \max_{p \in [0,1]} \mathbb{E}_{\theta\sim \Gamma}[\Up(p,\theta)]$.
	Then
	\begin{align*}
		\mathbb{E}_{\theta\sim \Gamma}[\Up(p^*,\theta)]&
        \le \sum_{\theta \in \Theta_\epsilon}  \gamma_\theta  \Up(p^*,\theta)  + 2\beta n \epsilon\tag{\Cref{lem:valClose}}\\
		& = \max_{p \in \Pcal} \sum_{\theta \in \Theta_\epsilon}  \gamma_\theta  \Up(p,\theta)  + 2\beta n \epsilon,\tag{\Cref{thm:toFinite}}.
	\end{align*}
	
	Hence, using this inequality we can finally prove that
    \begin{align*}
        R_T&=T\cdot\mathbb{E}_{\theta\sim \Gamma}[\Up(p^*,\theta)]-\mathbb{E}\left[ \sum_{t \in T}  \mathbb{E}_{\theta\sim \Gamma}[\Up(p_t,\theta)]\right]\\
        &\le T\cdot\max_{p\in \Pcal} \sum_{\theta\in\Theta_\epsilon}\gamma_\theta \Up(p^*,\theta)-\mathbb{E}\left[ \sum_{t \in T}  \mathbb{E}_{\theta\sim \Gamma}[\Up(p_t,\theta)]\right]+2\beta n\epsilon T\\
        & = \widetilde O\left(\beta \cdot\poly(m,n) T^{3/4}\right) + 2 \beta n \epsilon T,
    \end{align*}
    which is $\widetilde O\left(\beta \cdot\poly(m,n) T^{3/4}\right)$, thus concluding the proof.
\end{proof}

\subsection{Sample Complexity of Single-Dimensional Contracts}
In this section, we show how our results can be extended to derive sample complexity bounds (see \Cref{sec:prelimLearning} for a definition of the problem).
%
%
Such extension is based on the observation that the regret minimizer by \citet{lattimore2020learning} that we used in \Cref{alg:reduction} actually guarantees a stronger result. In particular, $\PE$ guarantees that at each time $t$, there exists a subset of feasible arms $X_t\subseteq X$ such that with probability at least $1-\delta$ it holds:
\[
\mu(x)\ge \max_{x\in X}\mu(x)-O\left(\alpha\sqrt{d}+\sqrt{\frac{d}{t}\log(k/\delta)}\right) \quad \forall x\in X_t.
\]
Thus, after $\widetilde O\left(\frac{d\log(k/\delta)}{(\eta-c\cdot\alpha\sqrt{d})^2}\right)$ samples, where $c=O(1)$, any arm left in the set $X_t$ is approximately optimal (specifically, $\eta$-optimal). However, this only holds if $\eta\ge\Omega(\alpha\sqrt{d})$. Indeed, it is known that outside this regime, the problem requires samples that are either exponential in dimension or linear in the number of arms \citep{du2019good, lattimore2020learning}, essentially losing the advantage of having an underlying linear structure.


In our problem, we can control the misspecification error $\alpha$ by adjusting the parameter $\epsilon$, allowing us to ensure good sample complexity bounds for any given approximation level $\eta$. However, this is only possible if we have prior knowledge of an upper bound $\beta$ on the density. Therefore, the following result relies on a slight strengthening of \cref{ass:boundedDensity}, where we assume knowledge of such an upper bound $\beta$.
Setting $\epsilon=O\Big(\big(\frac{\eta}{2\beta n}\big)^2\Big)$, we get a number of samples which is $\widetilde O\left(\frac{\poly(\beta, n, m)}{\eta^4}{\log\left(\frac{1}\delta\right)}\right)$.
We have to make a minor adjustment to the algorithm of \citet{lattimore2020learning} to have simple stopping rules (our algorithm just stops after a fixed amount of samples, after which it returns a solution). Formally, we need to have any-time concentration bounds in order to handle the $\delta$ correctness of the algorithm.

\begin{restatable}{theorem}{samplecomplexity}\label{thm:sample}
    Under \Cref{ass:boundedDensity} and \Cref{ass:boundedContract}, for any given $\eta>0$ and $\delta\in(0,1]$ there exists an algorithm which finds a contract $p$ in $\widetilde O\left(\frac{\poly(\beta, n, m)}{\eta^4}{\log\left(\frac{1}\delta\right)}\right)$ samples such that
    \[
    \mathbb{P}\left[\mathbb{E}_{\theta\sim \Gamma}\left[\Up(p,\theta)\right]\ge \OPT-\eta\right]\ge 1-\delta.
    \]
\end{restatable}







%% file: src/conclusions.tex
\section{Discussion and Future Work}
In this work, we analyzed the complexity of contract design with single-dimensional types. Despite the recent reduction from multi-dimensional to single-dimensional types by \citet{castiglioni2025reduction} suggested that efficient algorithms might be unattainable, we demonstrate that the problem admits a PTAS when considering additive approximations. Moreover, we show that this result cannot be further improved as the problem does not admit a $\FPTAS$ if $\PolyClass\neq\NP$. This suggests that, to some extent, contract design in the single-dimensional setting is computationally easier than in the multi-dimensional case.
%


Beyond this computational perspective, an information-theoretic viewpoint also supports the idea that single-dimensional contract design is simpler than its multi-dimensional counterpart.
In an online learning setting, we derive a regret bound of $\widetilde O(\poly(\mathcal I)T^{3/4})$ and a sample complexity bound $\widetilde O\left(\frac{\poly(\mathcal I)}{\eta^4}\right)$, where $\eta$ represents the additive error and $\mathcal I$ denotes the instance size.
This stands in contrast to the multi-dimensional case, where regret bounds are either exponential in the instance size or nearly linear.

Several interesting directions remain open.
First, it would be interesting to extend our results to \emph{menus of contracts}. 
While we believe that our $\PTAS$ can be fairly easily extended to menus of deterministic contracts, we don't see a clear connection between our hardness result and the impossibility of designing an $\FPTAS$ for menus of contracts. 
Another promising direction is to explore the online setting further, providing both lower bounds (on regret or sample complexity) and positive results under weaker or alternative assumptions. 
A central question is whether the bounded density assumption is truly necessary, or other properties related to the single-dimensional structure of the problem can be exploited. 


%% file: src/appendix.tex
\section{Omitted Proofs from \Cref{sec:PTAS}: Existence of a $\PTAS$}

\lemmadiscretization*

\begin{proof}
	For the first part, we simply use the law of total expectation 
	\begin{align*}
		\mathbb{E}_{\theta\sim\Gamma} \left[ \Up(p, \rho(\theta))\right] &= \sum_{\Pcal_i\in\Pcal}\mathbb{E}_{\theta\sim\Gamma} \left[ \Up(p, \rho(\theta))\Big\vert\theta\in\Pcal_i\right]\mathbb{P}_{\theta\sim\Gamma}[\theta\in\Pcal_i]\\
		&=\sum_{\Pcal_i\in\Pcal}\Up(p, \rho(\theta_i))\mathbb{P}_{\theta\sim\Gamma}[\theta\in\Pcal_i]\\
		&=\sum_{\Pcal_i\in\Pcal}\gamma_{\theta_i}\Up(p, \rho(\theta_i)).
	\end{align*}
For the second part, take any $i\in[k]$ and $\theta,\tilde \theta\in \Pcal_i$. Then, for all actions $a\neq \rho(\theta)$ it holds:
	\begin{align*}
		\Ua_{\theta}(\rho(\theta),p)&:=\sum_{\omega\in\Omega} F_{\rho(\theta),\omega}p_\omega-{ \theta c_{\rho(\theta)}} \\
		&= \sum_{\omega\in\Omega} F_{\rho(\theta),\omega}p_\omega- \tilde\theta c_{\rho(\theta)}+(\tilde\theta-\theta)c_{\rho(\theta)}\\
		&\ge \sum_{\omega\in\Omega} F_{a,\omega}p_\omega- \tilde\theta c_{a}+(\tilde\theta-\theta)c_{b^{\tilde\theta}(p)}\tag{By definition $\Bcal^{\tilde \theta}(p)$}\\
		&=\sum_{\omega\in\Omega} F_{a,\omega}p_\omega-{ \theta c_{a}}+(\tilde\theta-\theta)(c_{b^{\tilde\theta}(p)}-c_{a})\\
		&\ge \sum_{\omega\in\Omega} F_{a,\omega}p_\omega-{ \theta c_{a}} -\diam(\Pcal_i)\tag{Cauchy-Schwarz}\\
		&=\Ua_{\theta}(p,a)-\diam(\Pcal_i).
	\end{align*}
	This implies $\rho(\theta)\in\Bcal_{\diam(\Pcal_i)}^\theta(p)$, concluding the proof.
\end{proof}

\section{Impossibility of a $\FPTAS$}\label{app:reduction}
In this section, we show the missing details from the reduction of \Cref{sec:hardness}.
\subsection{Missing proofs for \Cref{thm:reduction}: If Analysis}
\begin{lemma}\label{lem:BR_IF}
    Let $\Scal^\star\subset \Scal$ be a cover of $E$, and for each $S\in \Scal^\star$, set $p_{\omega_S}=\frac{1}{n}$.
	while, $p_{\omega^\star}=p_{\bar \omega}=0$ and $p_{\omega_S}=0$ for all $S\notin \Scal^\star$. Then the best response of any type $\theta_i$ for $i\in E$ is in the set $\{a_{i,S}\}_{S\in\Scal^\star}$.
\end{lemma}

\begin{proof}
        Let $\tilde a_i$ be any action in the set $\{a_{i,S}\}_{S\in\Scal^\star}$.
        The utility of playing an action $a_{j, S}$ (for $j\in S$ and $S\in\Scal$) is
	\begin{align*}
		\sum_{\omega \in \Omega} F_{ a_{j, S}, \omega } p_\omega - \theta_i \cdot c_{a_{j, S}} &= \sum_{S'\in\Scal^\star} F_{ a_{j, S}, {\omega_{S'}} } p_{\omega_{S'}} - \frac{i}{n} \cdot c_{a_{j, S}}\\
        & = F_{ a_{j, S}, {\omega_{S}} } p_{\omega_{S}} - \frac{i}{n} \cdot c_{a_{j, S}}  \\
        & =  F_{ a_{j, S}, {\omega_{S}} } p_{\omega_{S}} - \frac{i}{4j^2n}\mu  \\
        &= \mathbb{I}(S \in S^\star) \frac{1}{2in}\mu -  \frac{i}{4j^2n}\mu. 
	\end{align*}
 
    Hence, the agent's utility playing action $\tilde a_{i}$ is 
    \[
    \mathbb{I}(S \in S^\star) \frac{1}{2in}\mu-  \frac{i}{4i^2n}\mu = \frac{1}{2in}\mu-  \frac{1}{4in}\mu =   \frac{1}{4in}\mu.
    \]

    Now, we show that the utility playing any other action is not larger.
    Playing any other action $a_{j,S'}$ the agent's utility is at most
    \[
        \mathbb{I}(S' \in \Scal^\star) \frac{1}{2jn}\mu -  \frac{i}{4j^2n}\mu \le \frac{1}{2jn}\mu -  \frac{i}{4j^2n}\mu=
        \frac{1}{2jn}\mu \left(1-\frac{i}{2j}\right). 
    \]
    Then, observing that the function $j\mapsto \frac{1}{j}(1-\frac{i}{2j})$ is increasing for $j\le i$ and decreasing for $j\ge i$, we conclude that the maximum incurs at $j=i$ with utility $\frac{1}{4in}\mu$.
    Hence, the agent's utility playing $a_{j,S'}$ us at most
    \[ \frac{1}{2jn}\mu \left(1-\frac{i}{2j}\right) \le \frac{1}{4in}\mu. \]
	
   Now, we consider the agent's utility playing an action $\bar a_{j, S}$, $j\in E$ and $S\in\Scal^\star$ (we can consider only $S\in\Scal^\star$, since, if $S\notin\Scal^\star$, then the agent utility is strictly negative):
	\begin{align*}
		\sum_{\omega \in \Omega} F_{ \bar a_{j, S}, \omega } p_\omega - \frac{i}{n} \cdot c_{\bar a_{j, S}} &=  F_{ \bar a_{j, S}, {\omega_{S}}} p_{\omega_S}  - \frac{i}{n} \cdot c_{\bar a_{j, S}}\\
		&\le \frac{1}{n} F_{ \bar a_{j, S}, {\omega_{S}}} -\frac{i}{4j^2n}\mu(1-\eta) \\
		&=\frac{1}{2jn}\left(1-\frac\eta2\right)\mu-\frac{i}{4j^2n}\mu(1-\eta) \\
	&=\frac{1}{2jn}\mu\left(1-\frac{\eta}{2}-\frac{i}{2j}(1-\eta)\right).
	\end{align*}
	
	Then, we can observe that $j\mapsto \frac{1}{j}\left(1-\frac{\eta}{2}-\frac{i}{2j}(1-\eta)\right)$ is decreasing for $j\le 2i\frac{1-\eta}{2-\eta}$ and increasing otherwise. Since $2i\frac{1-\eta}{2-\eta}\in[i-1, i]$ for $\eta$ small enough, we only need to show that the agent's utility for actions $\bar a_{i, S}$ and $\bar a_{i-1, S}$ are suboptimal.
	Straightforward calculations shows that the value of action $\bar a_{i-1, S}$ is at most $\frac{1}{2(i-1)^2n}\mu\left[\frac{i}{2}+\frac{\eta}{2}-1
    \right]$, while the value of action $\bar a_{i, S}$ is $\frac{1}{4in}\mu$. Then, if $\eta$ is small enough it holds $\frac{1}{2(i-1)^2n}\mu\left[\frac{i}{2}+\frac{\eta}{2}-1
    \right] <\frac{1}{4in}\mu$. Thus $\bar a_{i, S}$ with $S\in\Scal^\star$ is optimal among all actions $\{\bar a_{j, S}\}_{j, S}$ and provides agent's utility $\frac{1}{4in}\mu$.
	
    Finally, consider action $a^\star$. The agent's utility playing this action would be
	\begin{align*}
		\sum_{\omega\in \Omega} F_{a^\star,\omega}p_\omega - \frac{i}{n}\cdot c_{a^\star}  = \frac{1}{n}\sum_{S\in \Scal^\star} F_{a^\star,\omega_S} - \frac{i}{n}
		\le \frac{\epsilon m}{n} - \frac{i}{n}
		\le 0.
	\end{align*}
    Hence, we can finally conclude that $\tilde a_i$ is a best response for type $\theta_i$.
\end{proof}

\begin{lemma}\label{lem:BR_IF2}
    Under the same conditions of \Cref{lem:BR_IF}, the best response of type $\theta_0$ is $a^\star$.
\end{lemma}
\begin{proof}
    An agent of type $\theta_0$ doesn't pay any cost for playing actions, and hence their best response simply maximizes the expected payment. 
    The expected payment playing $a^\star$ would be at least $\frac{k\epsilon}{n}$ which is greater than $\mu/n$. On the other hand, the value of playing action $a_{i, S}$ or $\bar a_{i,S}$ (with $i\in S$) would be at most $\frac{1}{2in}\mu$ and $\frac{1}{2in}\mu(1-\eta/2)$ which are both strictly less then $\mu/n$. Thus, we conclude that an agent of type $\theta_0$ will play action $a^\star$.
\end{proof}

\subsection{Missing Proofs for \Cref{thm:reduction}: Only if Analysis}

\begin{lemma}\label{lem:onlyif1}
   The principal utility for types $i\in E_1$ is at most $\frac{1}{2in}\mu+\frac{1}{i}p_{\omega^\star}\mu\left(\frac{2}{\eta}-1\right)$.
\end{lemma}
\begin{proof}
    	We start analyzing the constraints that the best responses of types $\theta_i$, $i \in E_1$, impose on the contract.
	Consider a type $\theta_i$, $i \in E_1$, and let $a_{i,S}$ be their best response. Notice that $i \in S$.
	Then, from the IC constraint with respect to action $\bar a_{i,S}$ it holds
	\[
	p_{\omega_{S}}\frac{1}{2i}\mu + p_{\omega^\star}  \frac{1}{i}\mu+\left(1-\frac{3}{2i}\mu\right)p_{\bar\omega}- \frac{1}{4in}\mu \ge  p_{\omega_{S}}\frac{1-\frac{\eta}{2}}{2i} \mu +\left(1-\frac{1-\frac{\eta}{2}}{2i}\mu\right)p_{\bar\omega}- \frac{1}{4in} (1-\eta) \mu,
	\]
	implying
	\begin{align}\label{eq:largeE1}
	    p_{\omega_{S}}\ge \frac{1}{n}-\frac{4}{\eta}p_{\omega^\star}.
	\end{align}
	From this inequality it follows that the principal utility when the agent's type is $\theta_i$, $i\in E_1$ is at most
	\begin{align*}
	&F_{a_{i,S},\omega_S}(r_{\omega_S}-p_{\omega_S})+F_{a_{i,S}}(r_{\omega^\star}-p_{\omega^\star})+F_{a_{i,S},\bar \omega}(r_{\bar \omega}-p_{\bar\omega})\nonumber\\
        &\hspace{6cm}\le-\frac{1}{2i}\mu(\frac{1}{n}-4p_{\omega^\star}/\eta)+\frac{1}{i}\mu(\frac{1}{n}-p_{\omega^\star})\nonumber\\
	&\hspace{6cm}=\frac{1}{2in}\mu+\frac{1}{i}p_{\omega^\star}\mu\left(\frac{2}{\eta}-1\right),\label{eq:boundE1}
	\end{align*}
    as desired.
\end{proof}

\begin{lemma}\label{lem:onlyif2}
   The principal utility for types $i\in E_2$ is at most $\mu\left(\frac{1}{2in}-\frac{1}{8n^4}+\frac{2}{\eta}p_{\omega^\star}\right)$.
\end{lemma}
\begin{proof}
    	Now, consider an agent's type $\theta_i$, $i\in E_2$, and let $a_{j,S}$, $j \neq i$ be the action played by the agent. 
	Notice that $j \in S$. Then, IC constraint with respect to action $\bar a_{j,S}$ reads
	\[
	p_{\omega_{S}} \frac{1}{2j} \mu + \frac{1}{j}p_{\omega^\star}\mu+(1-\frac{3}{2j}\mu)p_{\bar\omega}- \frac{i}{4j^2n}\mu \ge p_{\omega_{S}} \frac{1}{2j} (1-\eta/2) \mu+\left(1-\frac{1-\frac{\eta}{2}}{10j}\mu\right)p_{\bar\omega}- \frac{i}{4j^2n} (1-\eta) \mu,
	\]
	which implies that
	\[
	p_{\omega_{S}} \ge \frac{i}{jn}-  \frac{4}{\eta} p_{\omega^\star}. 
	\]
	
	Hence, the principal's expected utility when the realized type is a $\theta_i$, $i \in E_2$, is at most:
	\begin{align*}
		F_{a_{j,S},\omega_S}(r_{\omega_S}-p_{\omega_S})+F_{a_{j,S},\omega^\star}(r_{\omega^\star}-p_{\omega^\star})+F_{a_{j,S},\bar \omega}(r_{\bar \omega}-p_{\bar\omega})&\le-\frac{1}{2j}\mu\left( \frac{i}{jn}- p_{\omega^\star} \frac{4}{\eta}\right)+\frac{1}{jn}\mu\\
		&=\mu\left(\frac{1}{jn}-\frac{i}{2j^2n}+\frac{2}{j\eta}p_{\omega^\star}\right)\\
		&\le\mu\left(\frac{1}{jn}-\frac{i}{2j^2n}+\frac{2}{\eta}p_{\omega^\star}\right)
	\end{align*}

	Now consider the term $\frac{1}{j}-\frac{i}{2j^2}$.
    This term over $E$ is maximized by $j=i-1$ or $j=i+1$ (recall that we cannot have that $j=i$ as $i\in E_2$). 
    For $j=i-1$, we can obtain an upper bound on this quantity observing that
	\begin{align*}
		\frac{1}{(i-1)}-\frac{i}{2(i-1)^2}&=\frac{1}{2i}-\frac{1}{2(i-1)^2i}\\
		&\le \frac{1}{2i} -\frac{1}{2(n-1)^2n}\tag{$x\mapsto -\frac{1}{(x-1)^2x}$ is increasing for $x\ge 1$}\\
		&\le \frac{1}{2i}-\frac{1}{2n^3}
	\end{align*}
    For $j=i+1$, we can obtain an upper bound on this quantity observing that
	\begin{align*}
		\frac{1}{(i+1)}-\frac{i}{2(i+1)^2}&=\frac{1}{2i}-\frac{1}{2(i+1)^2i}\\
		&\le \frac{1}{2i} -\frac{1}{2(n+1)^2n}\tag{$x\mapsto -\frac{1}{(x+1)^2x}$ is increasing for $x\ge 1$}\\
		&\le\frac{1}{2i} -\frac{1}{8n^3}
	\end{align*}

	Hence, the principal's expected utility when the realized type is a $\theta=i \in E_2$ is at most:
	\begin{align*}\label{eq:boundE2}
	\mu\left(\frac{1}{2in}-\frac{1}{8n^4}+\frac{2}{\eta}p_{\omega^\star}\right),
	\end{align*}
    as stated.    
\end{proof}

\begin{lemma}\label{lem:onlyif3}
   The principal utility for types $i\in E_3$ is at most $0$.
\end{lemma}

\begin{proof}
    	Now consider any type $i\in E_3$. This agent plays an action $\bar a_{j, S}$ for some $S$ and $j\in S$, actions $a^\star$ or action $a_0$.
	Clearly when actions $\bar a_{j, S}$ and $a_0$ are played the principal's utility is at most zero, as the only outcomes reachable from these actions are $\omega_S$ and $\bar\omega$ which have reward $0$. Moreover, the principal utility is at most zero also for action $a^\star$. Indeed, consider the IC constraints of action $a^\star$ with respect to $a_0$:
	\[
	\sum_{S\in\Scal}F_{a^\star,\omega_S}p_{\omega_S}-\frac{i}{n} \cdot c_{a^\star} + F_{a^\star,\omega^\star} p_{\omega^\star}\ge 0,
	\]
	which implies that $\epsilon\sum_{S\in\Scal}p_{\omega_S}+(1-m\epsilon)p_{\omega^\star}\ge \frac{i}{n}$. Then, the principal utility, when action $a^\star$ is played is 
	\begin{align*}
		\sum_{S\in\Scal} F_{a^\star,\omega_S}(r_{\omega_S}-p_{\omega_S})+F_{a^\star,\omega^\star}(r_{\omega^\star}-p_{\omega^\star})=&-\epsilon\sum_{S\in\Scal}p_{\omega_S}+(1-m\epsilon)\left(\frac{1}{n}-p_{\omega^\star}\right)\\
		&\le \frac{1-m\epsilon}{n}- \frac{i}{n}.
	\end{align*}
	and thus, the principal utility is negative also for action $a^\star$.
    Hence, we can conclude that when the realized type is $\theta_i$, $i\in E_3$, the principal's utility is at most zero.
	%
\end{proof}

\begin{lemma}\label{lem:onlyif4}
    When the realized type is $\theta_0$ the principal utility is $-\epsilon\sum_{S\in\Scal}p_{\omega_S}+(1-m\epsilon)(\frac{1}{n}-p_{\omega^\star})$.
\end{lemma}
\begin{proof}
    It is easy to see that type $\theta_0=0$ best responds with action $a^\star$  to any contract $p$ with expected principal's utility at least $\ell-n^{15}m^{-1}$. 
    Indeed, the maximum expected reward obtaining without playing $a^\star$ is obtained playing an action $a_{1,S}$ with expected reward
    \[ 
    \frac{1}{n} \mu \le \ell-n^{15}m^{-1}< \ell-n^{15}m^{-1}.  
    \]
    Hence, we conclude computing the principal's expected utility when the agent's type is $\theta_0$ playing action $a^\star$:
    
	\begin{align*}\label{eq:bound0}
	-\epsilon\sum_{S\in\Scal}p_{\omega_S}+(1-m\epsilon)\left(\frac{1}{n}-p_{\omega^\star}\right).
	\end{align*}
    as desired.
\end{proof}

\begin{lemma}\label{lem:onlyif5}
Define 	$\bar \Scal:=\{S\in\Scal: p_{\omega_S}\ge\frac{1}{n}-\frac{4}{\eta}p_{\omega^\star}\}$, then the principal utility is at most
    \[
        \frac{1-\rho}{n}\mu\left[\sum_{i\in E_1}\frac{1}{2in}+\sum_{i\in E_2}\left(\frac{1}{2in}-\frac{1}{8n^4}\right)\right]+\frac{1}{n}\rho(1-m\epsilon)-\frac{1}{n}\epsilon\rho|\bar\Scal|.
    \]
\end{lemma}

\begin{proof}
        Now, we combine the bound on the principal's utility for the different types.
    Let define the set $\bar \Scal:=\{S\in\Scal: p_{\omega_S}\ge\frac{1}{n}-\frac{4}{\eta}p_{\omega^\star}\}$. Notice that by \Cref{eq:largeE1}, for each element $i \in E_1$ it hold $i \in \bigcup_{S\in \bar \Scal} S$. Combining \Cref{lem:onlyif1}, \Cref{lem:onlyif2} and \Cref{lem:onlyif3}, we can upperbound the  principal's expected  utility with
    \begin{align}
    &\frac{1-\rho}{n}\mu\left[\sum_{i\in E_1}\right(\frac{1}{2in}+\frac{1}{i}p_{\omega^\star}(\frac{2}{\eta}-1)\left)+\sum_{i\in E_2}\left(\frac{1}{2in}-\frac{1}{8n^4}+\frac{2}{\eta}p_{\omega^\star}\right)\right]\nonumber\\
    &\hspace{6cm}+\rho(1-m\epsilon)\left(\frac{1}{n}-p_{\omega^\star}\right)-\epsilon\rho\sum_{S\in\Scal}p_{\omega_S}\nonumber\\
    & \le \frac{1-\rho}{n}\mu\left[\sum_{i\in E_1}\right(\frac{1}{2in}+\frac{1}{i}p_{\omega^\star}(\frac{2}{\eta}-1)\left)+\sum_{i\in E_2}\left(\frac{1}{2in}-\frac{1}{8n^4}+\frac{2}{\eta}p_{\omega^\star}\right)\right]\nonumber\\
    &\hspace{6cm}+\rho(1-m\epsilon)\left(\frac{1}{n}-p_{\omega^\star}\right)-\epsilon\rho |\bar\Scal|\left(\frac{1}{n}-\frac{4}{\eta}p_{\omega^\star}\right) \label{eq:upperbound}
    \end{align}
    where in the inequality we used that $\sum_{S\in\Scal} p_{\omega_S}\ge |\bar\Scal|(\frac{1}{n}-\frac{4}{\eta}p_{\omega^\star})$ by the definition of $\bar \Scal$.
    Then, we show that this upperbound on the principal's utility is maximized for $p_{\omega^\star}=0$. Indeed, in the expression the coefficient of $p_{\omega^\star}$ is 
	\[
	\frac{1-\rho}{n}\mu\left[\sum_{i\in E_1}\frac{1}{i}\left(\frac{2}{\eta}-1\right)+\sum_{i\in E_2}\frac{2}{\eta}\right]-\rho(1-m\epsilon)+\frac{4}{\eta}\epsilon\rho|\bar\Scal|.
	\]
	The following calculations shows that this is negative, implying that the expression is maximized for $p_{\omega^\star}=0$:
	\begin{align*}
		\frac{1-\rho}{n}\mu&\left[\sum_{i\in E_1}\frac{1}{i}\left(\frac{2}{\eta}-1\right)+\sum_{i\in E_2}\frac{2}{\eta}\right]-\rho(1-m\epsilon)+\frac{4}{\eta}\epsilon\rho|\bar\Scal|\\
		&\le \frac{1-\rho}{n}\mu\sum_{i\in E_1\cup E_2}\frac{2}{\eta} - \rho(1-m\epsilon)+4\frac{\rho \epsilon m}{\eta}\\
		&\le \frac{1-\rho}{n}\mu\frac{2}{\eta}n-\rho(1-m\epsilon)+4\frac{\rho \epsilon m}{\eta}\\
		&\le \frac{2}{\eta}\mu-\rho(1-m\epsilon)+4\frac{\rho \epsilon m}{\eta}\\
		&\le 2\frac{\mu}{\eta}-\frac{\rho}{2}+4\frac{\rho \epsilon m}{\eta}\\
		&=2\frac{n^2}{n^9m}-\frac{1}{2n^6}+4\frac{mn^2}{n^6n^8m}\\
		&=2\frac{1}{n^7}-\frac{1}{2n^6}+4\frac{1}{n^{12}},
	\end{align*}
	
	which is negative for $n\ge 2$.
	
	Thus, setting $p_{\omega^\star}=0$ in \Cref{eq:upperbound}, the principal utility is upper bounded by
	\[
	\frac{1-\rho}{n}\mu\left[\sum_{i\in E_1}\frac{1}{2in}+\sum_{i\in E_2}\left(\frac{1}{2in}-\frac{1}{8n^4}\right)\right]+\frac{1}{n}\rho(1-m\epsilon)-\frac{1}{n}\epsilon\rho|\bar\Scal|,
	\]
    concluding the proof.
\end{proof}

\section{Omitted Proofs from \Cref{sec:learning}: Learning Optimal Contracts}

\lemmaClose*

\begin{proof}
	Consider any contract $p\in[0,1]^m$. Then, for each action $a\in A$, let $\Theta(p,a)= \{\theta \in \Theta: b^{\theta}(p)=a\}$ which is the set of types that play $a$ as a best-response to $p$. 
    Let $\widehat \Theta(p)\subseteq \Theta_\epsilon$ be the set of $\theta \in \Theta_\epsilon$ such that in the interval $(\theta-\epsilon/2, \theta+\epsilon/2)$ there are multiple best-responses, \ie there are two actions $a\neq a'$ such that $\Theta(p,a)\cap (\theta-\epsilon/2,\theta+\epsilon/2]\neq\emptyset$ and $\Theta(p,a') \cap (\theta-\epsilon/2,\theta+\epsilon/2]\neq \emptyset$.
	Then, consider any $\theta\in\widehat\Theta(p)$ and observe that by \Cref{ass:boundedDensity} it holds
    \begin{align*}
        \left| \gamma_\theta \Up(p,\theta)-\int_{\theta-\epsilon/2}^{\theta+\epsilon/2} f_\Gamma(\tilde \theta) \Up(p,\tilde \theta)d\tilde\theta \right|&= \left| \int_{\theta-\epsilon/2}^{\theta+\epsilon/2} f_\Gamma(\tilde \theta) \left(\Up(p, \theta)-\Up(p,\tilde \theta)\right)d\tilde\theta \right|\\
        &\le \int_{\theta-\epsilon/2}^{\theta+\epsilon/2}f_\Gamma(\tilde \theta) \left|\Up(p, \theta)-\Up(p,\tilde \theta)\right|d\tilde\theta\\
        &\le 2\beta\epsilon.
    \end{align*}
    On the other hand, for any type $\theta\in\Theta_\epsilon\setminus\widehat\Theta(p)$ we have by definition that all the types in the interval $(\theta-\epsilon/2,\theta+\epsilon/2]$ have the same best response and thus
    \[
    \gamma_\theta\Up(p,\theta)=\int_{\theta-\epsilon/2}^{\theta+\epsilon/2} f_\Gamma(\tilde\theta) \Up(p,\tilde\theta)d\tilde\theta.
    \]
    Combining the two, we get that
    \begin{align*}
        \left|\mathbb{E}_{\theta\sim\Gamma}[\Up(p,\theta)]-\sum_{\theta\in\Theta_\epsilon}\gamma_\theta \Up(p,\theta)\right|&\le \sum_{\theta\in\Theta_\epsilon}\left|\gamma_\theta\Up(p,\theta)-\int\limits_{\theta-\epsilon/2}^{\theta+\epsilon/2}f_\Gamma(\tilde\theta)\Up(p,\tilde\theta)d\tilde\theta\right|\\
        &\le \sum_{\theta\in\widehat\Theta(p)}\left|\gamma_\theta\Up(p,\theta)-\int\limits_{\theta-\epsilon/2}^{\theta+\epsilon/2}f_\Gamma(\tilde\theta)\Up(p,\tilde\theta)d\tilde\theta\right|\\
        &\le 2|\widehat\Theta(p)|\beta\epsilon.
    \end{align*}

    Now we argue that $|\widehat\Theta(p)|$ is at most $n$. It is well known~\citep{alon2021contracts} that $\Theta(p,a)$ is an interval for each $a \in A$, and thus the number of intersections between such intervals is at most $n-1$, which is on to itself an upper-bound on $|\widehat\Theta(p)|$.
\end{proof}

\subsection{Proof of \Cref{thm:sample}}

\samplecomplexity*

\begin{proof}
    The proof essentially delves into the analysis of regret guarantees of $\PE$ presented in \citet{lattimore2020learning}. The algorithm $\PE$ works by partitioning the time horizon into geometrically increasing blocks each of length $T_\ell=2^{\ell-1}\lceil 4d\log\log d+16\rceil$. In each window each arm $x\in X$ is pulled $T_\ell\cdot\rho(x)$ times where $\rho\in\Delta_X$ is the optimal design, which has support at most $4d \log \log d + 16$ \citep{todd2016minimum}. The rewards observed are used to compute an estimate $\widehat \phi_\ell$ of $\phi$ as
    \[
    \widehat \phi_\ell=G_\ell^{-1}\sum_{k=T_{\ell}+1}^{T_{\ell+1}} \mu_k e_{x_t}\quad\text{where}\quad G_\ell=\sum_{x\in X} \lceil m_\ell\rho(x)\rceil x_tx_t^\top,
    \]
    and where $e_x\in\Reals^{k}$ is the vector with all zeros except the component associated with the arm $x$ (according to whatever fixed ordering).
    Then $\widehat \phi_\ell$ is used to discard provably suboptimal arms and maintain a set of optimal arms $X_t$ is updated at $T_l$
    \[
    X_{\ell}=\left\{x\in X_{{\ell-1}}:\max_{y\in X_{{\ell-1}}} (y-x)^\top \widehat\phi_\ell\le2\sqrt{\frac{4d}{T_\ell}\log\left(\frac{k}{\delta_\ell}\right)}\right\},
    \]
    where $\delta_\ell$ will be set in the following.
    Now we can use that key claim that with high probability $\widehat \phi_\ell$ is a good estimate of $\phi$.
    \begin{claim}[{\cite[Appendix~D]{lattimore2020learning}}]\label{claim:estimation}
    For any $\delta_\ell\in(0,1)$, the event 
    \[
    \Ecal_\ell=\left\{|x^\top(\widehat \phi_\ell-\phi)|\le 2\alpha\sqrt{d}+\sqrt{\frac{4d}{T_\ell}\log\left(\frac{k}{\delta_\ell}\right)}\quad\forall x\in X\right\}
    \]
    holds with probability at least $1-\delta_\ell$.
    \end{claim}
    Now define $x^\star=\max_{x\in X}\mu(x)$ and $\widehat x_\ell=\max_{x\in X_\ell}x^\top\widehat \phi_\ell$ and consider now any arm $x\in X_\ell$ the following inequalities
    \begin{align}
        2\sqrt{\frac{4d}{T_\ell}\log\left(\frac{k}{\delta_\ell}\right)}&\ge (\widehat x_\ell-x)^\top \widehat \phi_\ell\tag{Since $x\in X_\ell$}\\
        &\ge (x^\star-x)^\top \widehat \phi_\ell\tag{Definition of $\widehat x_\ell$}\\
        &\ge (x^\star-x)^\top \phi - 4\alpha\sqrt{d}-2\sqrt{\frac{4d}{T_\ell}\log\left(\frac{k}{\delta_\ell}\right)}\tag{Under event $\Ecal_\ell$}
    \end{align}

    Now, we can use the equation above to obtain a bound on the sub-optimality of any arm $x\in X_\ell$:
    \begin{align*}
        \mu(x)&\ge x^\top \phi-\alpha\\
        &\ge x^{\star,\top} \phi-\alpha- 4\alpha\sqrt{d}-4\sqrt{\frac{4d}{T_\ell}\log\left(\frac{k}{\delta_\ell}\right)}\\
        &\ge \mu(x^\star)-2\alpha- 4\alpha\sqrt{d}-4\sqrt{\frac{4d}{T_\ell}\log\left(\frac{k}{\delta_\ell}\right)}
    \end{align*}

    The event $\Ecal=\cup_{\ell=1}^{+\infty}\Ecal_{\ell}$ holds with probability at least $1-\sum_{\ell=1}^{+\infty}\delta_\ell$. Thus, by fixing $\delta_\ell = \frac{6\delta}{\pi^2\ell^2}$, we have that $\Ecal$ holds with probability at least $1-\delta$.

    Now we need to find the best lowest $\ell$ such that $\mu(x)\ge \mu(x^\star)-\eta$, or equivalently such that
    \begin{equation}\label{eq:tmp7}    
    2\alpha+ 4\alpha\sqrt{d}+4\sqrt{\frac{4d}{T_\ell}\log\left(\frac{k\pi^2\ell^2}{6\delta}\right)}\le \eta
    \end{equation}
    which is implied by the more stringent inequality
    \(
    6\alpha\sqrt{d}+4\sqrt{\frac{1}{2^\ell}\log\left(\frac{k\pi^2\ell^2}{6\delta}\right)}\le \eta.
    \)
    Now, we can use a standard technical claim (\emph{e.g.}, similar to \citet[Lemma~12]{jonsson2020planning}) that we can use to solve the equation above.
    \begin{claim}
        Define $D(c,\ell)=\sqrt{\frac{\log(c\cdot\ell^2)}{2^\ell}}$ for some $c>0$. Given any $z>0$, if $\ell\ge\log\left(\frac{2}{z^2}\log\left(\frac{c}{z^2}\right)\right)$ then $D(c,\ell)\le z$.
    \end{claim}
 
    We can use the claim above with $z=\frac{\eta-6\alpha\sqrt{d}}{4}$ and $c=\frac{k\pi^2}{6\delta}$ and thus we known that if 
    \[
    \ell\ge \ell^\star\coloneqq\log\left(\frac{32}{(\eta-6\alpha\sqrt{d})^2}\log\left(\frac{8k\pi^2}{3\delta (\eta-6\alpha\sqrt{d})^2}\right)\right)
    \]
    then \Cref{eq:tmp7} is satisfied. If we are going to use $\ell^\star$ blocks then we are actually using
    $\sum_{\ell=1}^{\ell^\star} T_\ell$ samples.
    \begin{align*}
        \sum_{\ell=1}^{\ell^\star} T_\ell &=\lceil4d\log\log d+16\rceil\sum_{\ell=1}^{\ell^\star}2^{\ell-1}\\
        &\le\lceil4d\log\log d+16\rceil 2^{\ell^\star}\\
        &\le\lceil4d\log\log d+16\rceil\frac{32}{(\eta-6\alpha\sqrt{d})^2}\log\left(\frac{8k\pi^2}{3\delta (\eta-6\alpha\sqrt{d})^2}\right).
    \end{align*}
    This proves that the sample complexity of misspecified linear bandits is $\tilde O\left(\frac{d}{(\eta-6\alpha\sqrt{d})^2}\log\left(\frac{k}{\delta}\right)\right)$.

    %
    When we reduce from our problem of contract design we can thus choose appropriately the discretization $\epsilon$. As per \Cref{lem:valClose}, for a discretization $\epsilon$ we have a misspecification of $\alpha = 2\beta n \epsilon$, a number of actions which is $k=|\Pcal|=\poly((n,m,1/\epsilon)^m)$ and a dimension of $d=|\Theta_\epsilon|=1/\epsilon$. By setting $\epsilon$ such that $6\cdot 2\beta n\epsilon\sqrt{d}=6\cdot 2\beta n\sqrt{\epsilon}=\frac\eta2$ we get $\epsilon = (\frac{\eta}{24\beta n})^2$, and we obtain a sample complexity of 
    \[
        C\frac{(\beta n)^2m\log\log(\beta n/\eta)}{\eta^4}\log\left(\frac{\poly(m,n)}{\delta}\right),
    \]
    which is $\tilde O\left(\frac{\poly(\beta,n,m)}{\eta^4}\log\left(\frac1\delta\right)\right)$ as stated.
\end{proof}